\documentclass[11pt]{scrartcl}
\usepackage[a4paper, total={16cm, 24cm}]{geometry}
\usepackage{natbib}
\renewcommand\tableofcontents{\listoftoc*{toc}}

\usepackage{authblk}

\makeatletter
\renewcommand\AB@affilsepx{, \protect\Affilfont}
\makeatother

\renewcommand\Affilfont{\large}
\author[1]{Niclas Boehmer}
\author[2]{Markus Brill}
\author[3]{Alfonso Cevallos}
\author[3]{Jonas Gehrlein} 
\author[4]{Luis S\'anchez-Fern\'andez}
\author[5]{Ulrike Schmidt-Kraepelin}
\affil[1]{Harvard University} 
\affil[2]{University of Warwick} 
\affil[3]{Web3 Foundation} 
\affil[4]{Universidad Carlos III de Madrid}
\affil[5]{TU Eindhoven}

\title{Approval-Based Committee Voting in Practice: \\
A Case Study of (Over-)Representation in the Polkadot Blockchain}

\usepackage[pagebackref]{hyperref}
\hypersetup{
	pdfencoding=auto, 
	psdextra,
	colorlinks=true,
	citecolor=green!50!black,
	linkcolor=red!60!black,
}

\usepackage{graphicx} 
\usepackage{natbib}  
\usepackage{caption} 
\usepackage{algorithm}
\usepackage{algorithmic}
\usepackage{subcaption}
\usepackage{multirow,makecell}
\usepackage{amsmath}
\usepackage{amsthm}
\DeclareMathOperator*{\argmax}{arg\,max}
\DeclareMathOperator*{\argmin}{arg\,min}
\usepackage{newfloat}
\usepackage{listings}
\DeclareCaptionStyle{ruled}{labelfont=normalfont,labelsep=colon,strut=off} 
\floatstyle{ruled}
\newfloat{listing}{tb}{lst}{}
\floatname{listing}{Listing}

\usepackage{cleveref}
\usepackage{booktabs}
\usepackage{amsfonts}
\usepackage{placeins}

\usepackage[table]{xcolor}
\usepackage{pgfplotstable}

\pgfplotstableset{
    color cells/.style={
        col sep=comma,
        string type,
        postproc cell content/.code={%
                \pgfkeysalso{@cell content=\rule{0cm}{2.4ex}\cellcolor{black!##1}\pgfmathtruncatemacro\number{##1}\ifnum\number>50\color{white}\fi##1}%
                },
        columns/x/.style={
            column name={},
            postproc cell content/.code={}
        }
    }
}

\usepackage{pgfplots}
\usepackage{tikz}

\usepackage[textsize=tiny]{todonotes}
\newcommand{\usk}[1]{}
\newcommand{\mb}[1]{}
\newcommand{\lsf}[1]{}
\newcommand{\jg}[1]{}

\newcommand{\toPol}[1]{}
\usepackage{siunitx}

\pgfplotsset{compat = 1.3}

\date{\vspace{-1cm}}

\newtheorem{definition}{Definition}
\newtheorem{theorem}{Theorem}
\newtheorem{proposition}{Proposition}

\begin{document}

\newpage

\maketitle

\begin{abstract}
\begin{center}
		\textbf{\textsf{Abstract}} \smallskip
	\end{center}
We provide the first large-scale data collection of real-world approval-based committee elections. These elections have been conducted on the Polkadot blockchain as part of their \textit{Nominated Proof-of-Stake} mechanism and contain around one thousand candidates and tens of thousands of (weighted) voters each. We conduct an in-depth study of application-relevant questions, including a quantitative and qualitative analysis of the outcomes returned by different voting rules. Besides considering proportionality measures that are standard in the multiwinner voting literature, we pay particular attention to less-studied measures of overrepresentation, as these are closely related to the security of the Polkadot network. We also analyze how different design decisions such as the committee size affect the examined measures. 
\end{abstract}

\vspace{13pt}
\hrule
\vspace{8pt}
{\small\tableofcontents}
\vspace{20pt}
\hrule

\section{Introduction}
Approval-based committee (ABC) voting describes the task of selecting a subset of candidates based on approval-style preferences of voters over candidates. 
A central concern in such elections is to ensure that voters' opinions are ``proportionally'' reflected in the selected set of candidates, i.e., in the \emph{committee}. 
Formally capturing proportional representation as axioms, and designing rules that are guaranteed to satisfy these axioms, is a very active area of research; see the book by \citet{LaSk22a}. 
While this effort has resulted in a good theoretical understanding of the axiomatic aspects of different voting rules, there are only very few studies that examine the actual behavior of such rules on specific real-world or synthetically generated 
voting instances \citep{EFL+17a,SFJ+22a,mapofRules,DBLP:journals/corr/abs-2305-08970}.
Nevertheless, these few empirical works have already proved useful, as they found that different voting rules often produce similar outcomes and that most voting rules tend to significantly outperform their worst-case proportionality guarantees. 
This already motivated the development of new proportionality axioms \citep{SkGo22a,BrPe23a}. 
Nevertheless, as proportionality axioms seem to lose their discriminative power in practice, the problem arises of how to measure the proportionality of outcomes and how to quantify the differences between proportional rules. 

One reason for the shortage of empirical works might be the lack of real-world data. 
Accordingly, previous empirical works often either resorted to synthetically generated data or converted data from other voting applications such as ordinal elections or participatory budgeting to fit the ABC voting setting.
Despite the fact that previous research has named a multitude of potential applications of ABC voting ranging from political elections \citep{BLS18a} to recommender systems \citep{DBLP:conf/recsys/StreviniotisC22,DBLP:conf/prima/StreviniotisC22,DBLP:conf/eumas/GawronF22} 
to forest management \citep{PBSH20a}, there are only very few applications where ABC elections have been implemented. 
A  so-far mostly unexplored exception are blockchain protocols that conduct ABC elections on a day-to-day basis. 
Specifically, these elections occur in blockchains using the Nominated Proof-of-Stake (NPoS) protocol. In this system, a subset of stakeholders, called \textit{validators}, 
are elected to run the consensus protocol, which is crucial for the integrity of the blockchain. 
The problem of selecting the validators can be modeled as an ABC voting problem, and, indeed, on the \textit{Polkadot} network (\url{https://polkadot.network}),  a proportional ABC voting rule is used \citep{polkadot-overview,CeSt21a}.\footnote{Polkadot currently uses Phragmén's sequential rule, but considers switching to the Phragmms rule \citep{CeSt21a}.}

\subsection{Our Contributions}
We complement the mostly theoretical literature on ABC voting in various directions, thereby contributing tools and insights for future empirical works. 
\begin{itemize}
    \item  We compile the first collection of real-world ABC elections consisting of $496$ 
    elections from the Polkadot blockchain. These elections  contain between \num{18202} and \num{48025} voters and between \num{920} and \num{1080} candidates.
    \item We conduct an in-depth study of the behavior of different voting rules, with a particular focus on their application in Polkadot. We compare the outcomes and analyze whether outcomes under- or overrepresent voters.
    Notably, security concerns in Polkadot provide a novel view on (and motivation for) the goal of preventing overrepresentation. 
    Our empirical findings, summarized in \Cref{sec:conc}, 
    contribute to the ongoing discussion in voting and blockchain research regarding the selection of voting rules. These insights offer compelling justifications for the use of more sophisticated, proportional voting rules, similar to those implemented in Polkadot.
    \item As part of our analysis, we initiate the study of quantitative measures for over- and underrepresentation  
    that are needed to distinguish and describe the behavior of voting rules on real-world data. We do so by adapting known proportionality axioms and by introducing a new measure regarding the prevention of overrepresentation. 
    \item We consider design decisions 
    concerning the chosen committee size and the question whether (copies of) the same candidate can be selected multiple times. 
    We make recommendations in light of our formulated desiderata.
\end{itemize}

\noindent All collected elections and the code for our experiments are available at \url{github.com/n-boehmer/ABC-practice-Polkadot}.
Another blockchain network that follows the NPoS protocol is \textit{Kusama}  (\url{https://kusama.network}). In Appendix~\ref{kusama}, we present $1520$ elections conducted in the Kusama network, with roughly \num{2000} candidates and \num{10000} voters each, and verify that most of our empirical findings also hold for the Kusama elections.

\subsection{Related Work}
There is a large literature that develops  (axiomatic) measures to assess the proportionality of a committee \citep{ABC+16a, SFF+17a, LaSk22a, BrPe23a,Skow21a, PPS21a}.
Questions typically examined in \textit{empirical} works on ABC voting concern the similarity of outcomes returned by different voting rules \citep{reichertElkind,mapofRules,EFL+17a}, how often voting rules satisfy certain proportionality axioms \citep{mapofRules,DBLP:journals/corr/abs-2305-08970}, and how often such axioms are satisfied by randomly selected committees \citep{BFNK19a,BrPe23a}. 
Complementary to our work, \citet{SFJ+22a} also provided tools for more empirical studies of ABC rules by proposing several synthetic models for generating data as well as a framework for visualizing the data and experimental results.
On a more general note, empirical works are more common in the context of ordinal voting. In terms of methodological contributions, the \mbox{\textit{PrefLib}} \citep{DBLP:conf/aldt/MatteiW13,trendspreflib} and \mbox{\textit{Pabulib}} \citep{DBLP:conf/ijcai/FaliszewskiFPP023}  databases, which contain large collections of real-world preference and participatory budgeting data and the recently developed map of elections framework \citep{DBLP:conf/atal/SzufaFSST20,DBLP:conf/ijcai/BoehmerBFNS21} are notable projects.

Regarding the usage of voting in blockchains, there are also some examples beyond Polkadot and Kusama. 
However, most of them use non-proportional voting rules. 
For instance, the EOS network uses the voting rule \emph{approval voting} \citep{grigg2017eos}, where the candidates with the highest number of approvals get selected. Notably, the usage of this rule has led to a series of complaints regarding centralization issues  \citep{downgrade,garg2019eos} 
(arguably related to the non-proportionality of this rule).  

\section{Preliminaries} \label{sec:preliminaries}
In this section, we formally introduce approval-based committee (ABC) elections and describe \emph{Polkadot}, a blockchain network that carries out such elections every day and serves as the main source of data for this paper.
Due to the general homogeneity of our data observed in \Cref{sec:description}, in our following analysis we mostly report average values, as the variance is usually quite low.

\subsection{Approval-Based Committee Voting}
\label{sec:abc}

For any $n\in \mathbb{N}$, we define $[n]=\{1,\dots, n\}$. A \textit{(weighted) ABC election} $E=(C,V,A,w,k)$ is defined by a set $C=\{c_1,\dots,c_m\}$ of candidates, a set $V=[n]$ of voters, an \emph{approval profile} $A$, a \emph{weight function} $w:V\mapsto [0,1]$ that maps each voter to its \emph{voting weight}, and the size~$k$ of the committee to be selected, where a \textit{committee} is a subset of the candidates. 
The approval profile $A$  consists of a subset $A_v\subseteq C$ for every voter $v\in V$, containing all candidates $v$ approves of. We allow voters to have different \emph{voting weights} (as in Polkadot), which are captured by the weight function $w$. Without loss of generality, we assume that $\sum_{v\in V} w(v)=1$.
For a subset $V'\subseteq V$ of voters, we let $w(V'):=\sum_{v\in V'} w(v)$ denote the sum of their voting weights. 
For a candidate $c\in C$, we let $V_c$ be the set of \emph{supporters} of $c$, i.e., $V_c=\{v\in V : c\in A_v \}$.
For a candidate $c\in C$, we define $w(c):=w(V_c)$ to be the \emph{approval weight} of $c$.
Extending this notation, for a set $C'\subseteq C$ of candidates their approval weight $w(C'):=\sum_{v\in V:  A_v\cap C'\neq \emptyset} w(v)$ is the summed voting weight of voters approving at least one candidate from $C'$.
The average \emph{satisfaction} of a voter group $V'\subseteq V$ with a committee $W$ is $\frac{1}{w(V')}\sum_{v\in V'} w(v)\cdot |W\cap A_v|$.
Given a committee $W\subseteq C$, for each $\ell\in [k]$, and each non-selected candidate $c\in C\setminus W$, 
an \emph{$\ell$-supporting group (of $c$)}  is a subset of $c$'s supporters $V'\subseteq V_c$ with $w(V') \geq \frac{\ell}{k}$. 

When discussing the legitimacy of a committee $W \subseteq C$, we often use \emph{vote assignments} $\alpha: V \times W \rightarrow [0,1]$. The idea is that any candidate $c \in W$ needs to be backed by voters supporting them; however, voters should not be counted multiple times across different candidates. This leads to the following constraints of a vote assignment: (i) for any $v \in V$ and $c \in W$, $\alpha(v,c) >0$ implies that $c \in A_v$, and (ii) for all $v\in V$ it holds that $\sum_{c \in A_v \cap W} \alpha(v,c) \leq w(v)$. Regarding the second constraint, we often additionally assume (implicitly or explicitly) that $\sum_{c \in A_v \cap W} \alpha(v,c) = w(v)$ as long as $A_v \cap W \neq \emptyset$. For a candidate $c \in W$, we are interested in its \emph{backing weight} (w.r.t. $\alpha$), defined by $\sum_{v \in V_c} \alpha(v,c)$.

\newcommand{\corr}[1]{\tikz{\pgfmathsetmacro{\myperc}{1.05*(#1-250)+10}\node [transform shape, rounded corners=1pt,fill=blue!\myperc,inner sep=0, minimum width=22pt, minimum height=8pt] {#1}; }}

An \emph{ABC voting rule} takes as input an election $E=(C,V,A,w,k)$ and outputs a size-$k$ committee.
Below, we define the six rules that we study in this paper; for all six of them, the adaptation to weighted votes is straightforward. While the formulations of the rules allow for ties, we use lexicographic tie breaking whenever necessary, i.e., if two candidates can be added to the committee, we add the one with the smaller index.

\paragraph{Approval Voting (AV)} AV selects $k$ candidates with highest approval weight. 

\paragraph{Satisfaction Approval Voting (SAV)} 
The rule was introduced by \citet{BrKi14a}.
Under SAV, the score of a candidate $c\in C$ is defined by $\sum_{v\in V_c}  \frac{w(v)}{|A_v|}$. The rule then selects $k$ candidates with the highest score.

\smallskip
The next four rules work in a sequential fashion, i.e., they start with the empty committee and add candidates one by one until $k$ candidates have been selected. We often informally refer to them as the \emph{proportional rules}. 

\paragraph{Sequential Proportional Approval Voting (seq-PAV)}\phantom{x}
The rule was introduced by \citet{Thie95a} and analyzed in detail by \citet{Jans16a} and \citet{ABC+16a}. The \emph{PAV score} of a committee $W$ is $\mathrm{sc}(W)=\sum_{v\in V} w(v)\cdot (\sum_{i=1}^{|A_v\cap W|} \frac{1}{i})$.
The rule starts with the empty committee and in each round adds a candidate with maximum marginal contribution, i.e., $\argmax_{c\in C} \mathrm{sc}(W\cup \{c\})- \mathrm{sc}(W)$.

\paragraph{Sequential Phragmén (seq-Phragmén)} The rule was introduced by \citet{Phra94a} and analyzed in detail by \citet{Jans16a} and \citet{BFJL16a}. 
In this rule, voters earn virtual money at a constant speed proportional to their voting weight.
At the beginning, each voter has no money. As soon as the supporters of a candidate jointly own $1$ unit of money, the candidate is added to the committee and the money of all its supporters is reset to zero. 
All other voters keep their money and the process is repeated $k$ times. 

\paragraph{Method of Equal Shares (MES)} The rule was introduced by \citet{PeSk20b}.
At the beginning, each voter $v\in V$ has a budget $b_v$ of $k\cdot w(v)$ units of virtual money. 
A candidate $c\in C$ is called $q$-affordable for some $q\in [0,1]$ if the supporters of $c$ can together pay one unit of virtual money without any of them paying more than $q$ units, i.e., $\sum_{v\in V_c} \min(b_v,q)\geq 1$. 
In each round, a not yet selected candidate $c$ which is $q$-afforable for a minimum value of $q$ is added to the committee and for each voter $v\in V_c$ we set $b_v=\max(0,b_v-q)$. 
MES might produce outcomes containing less than $k$ candidates, in which case we complete the outcome by running seq-Phragmén. In this case, the budget of a voter at the end of the execution of MES is their  starting budget for seq-Phragmén. 

\paragraph{Phragmms}
The rule was introduced by \citet{CeSt21a} as a combination of seq-Phragmén and the maximin support method \cite{SFFB22b}.\footnote{We do not include the maximin support method in our analysis due to its prohibitive computational complexity.}
The overarching goal of the rule is to select a committee $W$ together with a vote assignment $\alpha$ that guarantees high backing weight for any candidate in $W$. The rule repeats the following two steps iteratively: (i) Given $(W,\alpha)$, compute a score for any candidate in $c \in C\setminus W$ that corresponds to the highest backing weight $t$ that can be given to $c$ without decreasing the backing weight of any candidate in $W$ below $t$ while only doing local changes to the vote assignment. 
Then, add the candidate with highest score to $W$. (ii) Compute a new vote assignment for $W$ that is \emph{balanced}.\footnote{A balanced vote assignment maximizes the sum of the backing weights of the candidates, while minimizing the sum of squared backing weights. For details, see \citet{CeSt21a}.} 
\medskip

We do not include the maximin support method in our analysis due to its prohibitive computational complexity.

\subsection{Application Background}
\label{sec:polkadot}

The Polkadot blockchain \citep{pollWhite,polkadot-overview} 
implements a variation of Proof-of-Stake (PoS) as its consensus mechanism to determine the addition of new blocks to the blockchain.
These systems rely on a restricted set of \textit{validators} who 
are granted the exclusive privilege to append new blocks, notably avoiding the reliance on energy-intensive computing power which characterizes Proof-of-Work (PoW) systems, thus making it an environmentally friendly alternative.
This is different from Proof-of-Work blockchains such as Bitcoin \citep{nakamoto2008bitcoin}, where everyone can propose new blocks. 
For the integrity of PoS networks, it is vital that validators adhere to established rules when creating new blocks.
Crucially, the network remains secure as long as fewer than one-third of these validators behave maliciously \citep{lamport2019byzantine}. 

Polkadot operates as a permissionless network, which means that everyone can become a validator candidate. To address the arising selection problem, the network allows token holders to act as voters, referred to as \textit{nominators}, and screen and evaluate the available candidates, which is a generally quite challenging task \citep{GMGK23a}.
In particular, each token holder can submit a ballot approving up to $16$ validator candidates. 
Aggregating these casted ballots, a committee of $297$ active validators is selected in each era (day) by Polkadot's \textit{Nominated Proof-of-Stake (NPoS)} election algorithm, which uses the seq-Phragmén rule. 

Contrasting with other ABC elections, elections held on the Polkadot network are characterized by three unique aspects: Firstly, a voter's voting weight corresponds to their \textit{stake}, that is, the aggregate of tokens in their possession. Secondly, if a voter endorses a candidate who subsequently gets selected, the voter's stake is held as collateral. If this chosen candidate breaches protocol\,---\,for instance, by proposing blocks that are against the rules\,---\,the voter's stake may be seized, an action known as \textit{slashing} (see \Cref{sec:over}).
Lastly, a voter continuously receives rewards (in the form of network tokens) for approving candidates who have been elected and diligently perform their duties. These latter two characteristics align the economic incentives of voters with the network's interests, ensuring that voters include only candidates regarded as trustworthy on their ballot.
 
In the following sections, we discuss several desiderata that the Polkadot designers 
formulated for their voting rule, including underrepresentation (\Cref{sec:under}) and overrepresentation (\Cref{sec:over}) concerns. There are also several other desiderata that are beyond the scope of the paper, such as the running time of the algorithm and verification concerns \citep{CeSt21a}.

\section{A First View on Instances and 
Committees}\label{sec:description}

We start by describing characteristics of our collected dataset and analyzing the overlap between committees that are selected by the voting rules mentioned in \Cref{sec:abc}. 

\subsection{Description of the Data} In Polkadot, every day is viewed as an \emph{era}, and one election is conducted at the end of each era, implying that our $496$ collected elections have a time-based ordering.\footnote{
This also makes the data suitable for testing models that cover collective decision making over time  \citep{DBLP:conf/aaai/Lackner20,DBLP:conf/atal/BoehmerN21}. 
Currently, such data is very rare even in the context of ordinal single-winner elections \citep{trendspreflib,DBLP:conf/atal/BoehmerS23}.}
In particular, the studied elections cover eras $398$ (July 5, 2021) through $1078$ (May 16, 2023) and are generated from openly accessible data maintained and distributed by the \textit{Web3 Foundation}. 
However, the data stream contains some gaps, i.e., in $185$ eras from the above interval, elections were not correctly stored. 
The committee size in the elections is $300$ and voters are only allowed to approve of $16$ candidates, i.e., 
$k=300$ and $|A_v|\le 16$ for all $v \in V$.\footnote{The actual committee size currently used by Polkadot is 297. We use $k=300$ for easier readability of our results.}
Notably, it is possible to map voters and candidates of different elections to each other, as they all provide unique IDs. 

In \Cref{fig:basic1}, we present the sizes of the collected elections.
Each election contains between \num{18202} and \num{48025} voters and between \num{920} and \num{1080} candidates.
Thus, the number of voters exhibited a stronger fluctuation over time than the number of candidates. 
Nevertheless, we see here that neither of the two parameters changes too much from one era to the next (the few ``jumps'' in the plots are mostly due to missing elections). 
In terms of the  average number of candidates a voter approves, there is a monotonic decrease from around $9.7$ to around $7.5$ over time. 

\begin{figure}
     \centering
   \includegraphics[width=6cm]{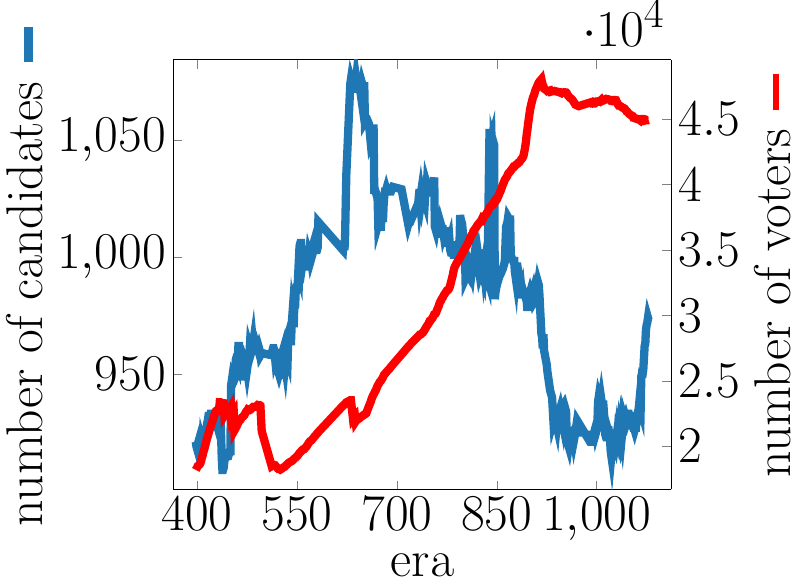}
  \caption{Number of candidates (blue) and voters (red) in our $496$ elections.}
\label{fig:basic1}
\end{figure}

\begin{table*}[t]
    \centering
\begin{tabular}{lcccccc}
\toprule
{} &       AV &      SAV &   seq-PAV &  Phragmms &  seq-Phrag. &      MES \\
\midrule
AV          &    $-$ &  266.190 &  263.583 &   266.772 &     262.595 &  262.163 \\
SAV         &  266.190 &    $-$ & 282.042 &   278.677 &      276.933 &  277.206 \\
seq-PAV      &  263.583 &  282.042 &    $-$ &   286.562 &      288.317 &  288.119 \\
Phragmms    &  266.772 &  278.677 &  286.562 &     $-$ &      288.091 &  286.887 \\
seq-Phrag. &  262.595 &  276.933 & 288.317 &   288.091 &        $-$ &  295.123 \\
MES        &  262.163 &  277.206 &  288.119 &   286.887 &      295.123 &    $-$ \\
\bottomrule
\end{tabular}
    \caption{Average overlap between committees returned by different voting rules. The committee size is $300$.}
    \label{tab:overlap}
\end{table*}

\subsubsection{How Much do Elections Change over Time?}\label{sub:datTime}

We take a closer look at how much elections change over time. 
For this, we make use of the fact that in Polkadot each voter and each candidate is identified by a stash address, which typically does not change over time. 
Accordingly, voters and candidates have the same identity in different elections. 

We examine four different quantities comparing two elections $E=(C,V,A,w,k)$ and $E'=(C',V',A',w',k)$. For each of them, the smaller the value, the smaller the change.
\begin{description}
    \item[Relative voter set change] This quantity captures how much the set of voters changed between the two elections: $\frac{|V\triangle V'|}{|V|+|V'|}$.
    \item[Relative weight change] This quantity captures how much the weight of voters that are present in both elections changes: $\frac{\sum_{v\in V\cap V'} |w(v)-w'(v)|}{\sum_{v\in V\cap V'} w(v)}$
    \item[Relative opinion change] This quantity captures how much the voters that are present in both elections change their opinions about the candidates that are also present in both elections. We take the average of the following expression over all voters $v\in V'\cap V$ (for two subsets of candidates: $\frac{|(A_v\cap C')\triangle (A'_v \cap C)|}{|(A_v\cap C')|+|(A'_v \cap C)|}$.
    \item[Relative candidate set change] This quantity captures how much the set of candidates changes: $\frac{|C\triangle C'|}{|C|+|C'|}$.
\end{description}

\begin{figure*}[t!]
\centering
\begin{subfigure}[b]{.24\textwidth}
  \centering
  \includegraphics[width=\textwidth]{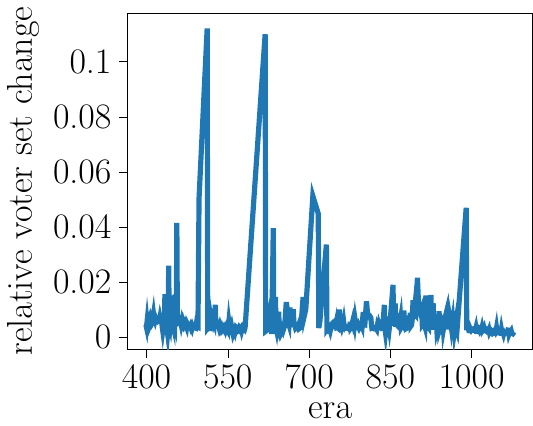}
  \caption{Relative voter set change.}
\end{subfigure}\hfill
\begin{subfigure}[b]{.24\textwidth}
  \centering
  \includegraphics[width=\textwidth]{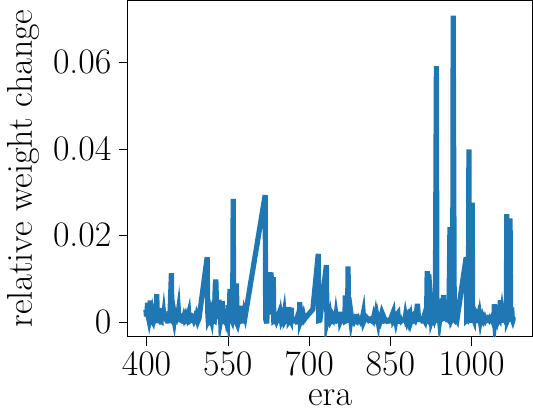}
  \caption{Relative weight change.}
\end{subfigure}\hfill
\begin{subfigure}[b]{.24\textwidth}
  \centering
  \includegraphics[width=\textwidth]{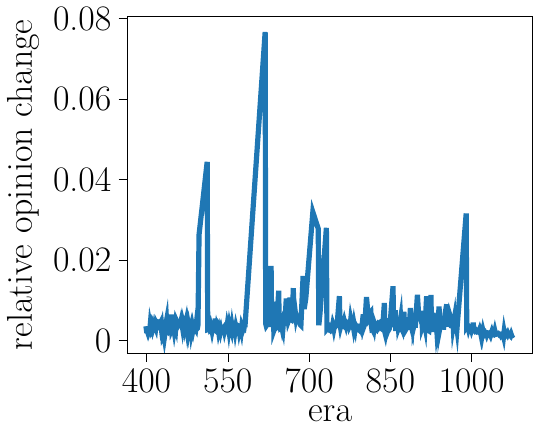}
  \caption{Relative opinion change.}
\end{subfigure} \hfill
\begin{subfigure}[b]{.24\textwidth}
  \centering
  \includegraphics[width=\textwidth]{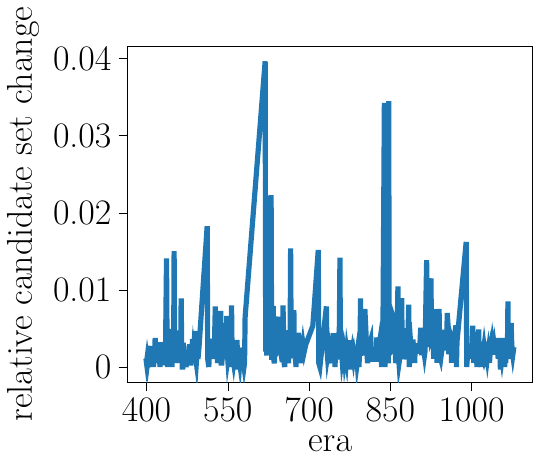}
  \caption{Relative candidate set change.}
\end{subfigure}
\caption{Some statistics capturing by how much elections change from one era to the next. }
\label{fig:basic}
\end{figure*}

\begin{figure*}[t]
    \centering
    \resizebox{0.33\textwidth}{!}{\input{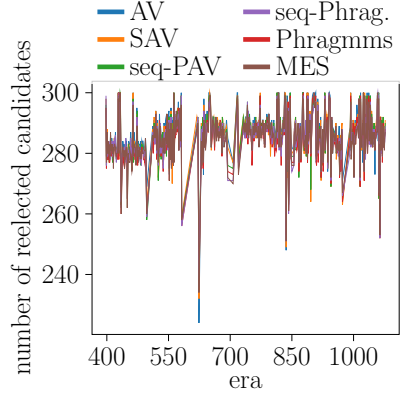}}
    \caption{Number of candidates that are also part of the winning committee in the next election.}
    \label{fig:changes}
\end{figure*}

For each of these four quantities, \Cref{fig:basic} depicts the change between two consecutive elections in our dataset. 
The picture for all of them is similar: Going from one election to the next, typically no large change occurs, which is also quite intuitive as without any active actions from participants, nothing will change.
Note that most of the spikes visible in \Cref{fig:basic} are due to gaps in our data. 
However, comparing the first election in our dataset to the last one, indeed quite some changes occurred: the relative voter set change is $0.76$, the relative candidate set change is $0.39$, the relative weight change is $0.4$, and the relative opinion change is $0.46$. 

\subsubsection{How Much do Committees Change Over Time?}\label{sub:comTime}

In \Cref{fig:changes}, we show the overlap between winning committees in two successive elections (recall that we can map candidates from different elections to each other by their stash address). In general, we observe that, in most cases, only around $20$ candidates get replaced when moving from one election to the next (the average value for all rules is between $13$ and $15$). 
Recalling that we have seen in \Cref{fig:basic} that elections change only marginally in many steps, this might still be viewed as surprisingly many changes (especially, recalling that the difference between the outcomes produced by the proportional rules on the same election is lower).
Beyond that, it is worth noting that the graphs show some eras with significantly smaller overlaps. 
Again, most of these are due to the gaps in our data, which naturally lead to  smaller overlap, as a larger period of time has passed between one election and its successors. 

Another trend that stands out from \Cref{fig:changes} is that the rules behave remarkably similarly, which means that all of them make roughly the same number of adjustments to account for changes in the election.  
This observation stands in partial contrast to previous theoretical and empirical results from the literature that different rules have a different sensitivity to random or adversarial changes in the election \citep{DBLP:journals/ai/BredereckFKNST21,BFJK23a,DBLP:conf/eaamo/BoehmerBFN22}. 
Yet, voters in Polkadot elections certainly do not change their votes randomly.

\subsubsection{Election Properties}\label{sub:prop}
In the following, we describe some structural properties of the Polkadot elections.
The plot in \Cref{fig:basic2} shows the averaged order statistics of the weight of voters: 
For each election, we sort the voter weights decreasingly and depict the averaged values here (we cut off the plot at the minimum number of voters in our elections). It stands out that the weight of voters is distributed very unevenly. On average, the top-$31$ voters combined have more weight than the remaining voters together (these powerful voters are often called ``whales'').
In \Cref{fig:basic3}, we repeat this analysis with the approval weight of candidates. 
Here, again, the distribution is far from uniform, yet slightly less imbalanced than for the voting weights: Some candidates have a substantially higher approval weight than others; however, typically around $300$ candidates have an approval weight of at least $0.01$. This mismatch between  \Cref{fig:basic2} and \Cref{fig:basic3} can be explained by the fact that voters can approve multiple candidates and that voters with high weights often have disjoint approval sets (voters with a high weight typically represent a company or organization which also adds some candidates in the election; accordingly each company lets their voters vote for their own candidates in order to ensure that they get selected and increase their monetary rewards). 

\begin{figure*}[t]
    \centering
  \begin{subfigure}[b]{.49\textwidth}
  \centering
  \resizebox{0.7\textwidth}{!}{\input{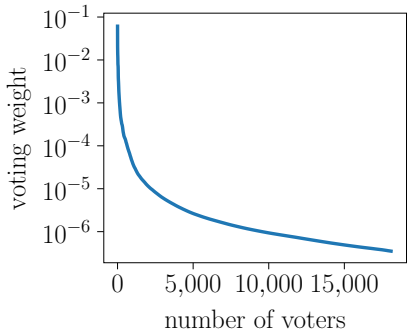}}
  \caption{Distribution of voting weight.}
  \label{fig:basic2}
\end{subfigure}\hfill
\begin{subfigure}[b]{.5\textwidth}
  \centering
  \resizebox{0.7\textwidth}{!}{\input{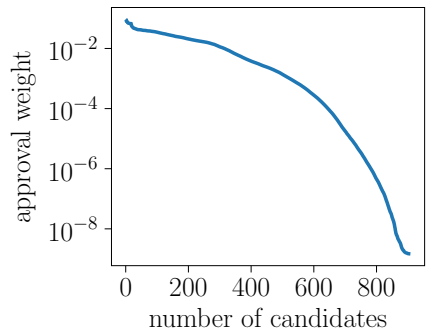}}
  \caption{Distribution: Candidate's approval~weight.}
  \label{fig:basic3}
\end{subfigure}
\caption{Some statistics regarding the $496$ collected Polkadot elections. The $y$-axes are logarithmic. }
\end{figure*}

\subsection{Overlap Between Committees}
In \Cref{tab:overlap}, we show the similarity between different voting rules in terms of the average overlap of their computed committees. 
We observe a high consensus of all rules, thereby confirming previous findings on synthetic data \citep{reichertElkind,mapofRules}; in fact, across all elections, two rules never disagree on more than $51$ out of the $k=300$ candidates. 
However, there are some differences: 
The four considered proportional rules 
have a particularly high agreement. 
This is most extreme for seq-Phragmén and MES, which have an average overlap of $295$ and never produce outcomes that differ in more than $11$ candidates. 
In contrast, AV returns committees that are furthest away from the outcomes of the other rules, while outcomes produced by SAV are in  general slightly closer to the ones returned by the proportional rules. 
Moreover, in Appendix~\ref{sub:comTime}, we analyze how much winning committees change over time and observe that, with respect to this aspect,  all our rules behave remarkably similarly to each other. In most cases, winning committees in successive time steps differ by at most $20$ candidates. 

\section{Preventing Underrepresentation}
\label{sec:under}
In Polkadot, voters are financially rewarded when (some of) their approved candidates appear in the selected committee.
Thus, it is important to ensure that no voter groups are underrepresented in the selected committee, as 
otherwise voters might feel disengaged and stop participating in the protocol. 
In general, unrepresentative outcomes can also lead to power centralization, which is very dangerous for the system.
In voting theory, underrepresentation is typically prevented by requiring proportional representation. 
Intuitively, this means that a group of voters with a summed voting weight of $\frac{\ell}{k}$ should be allowed to select $\ell$ of the committee members. 
In this section, we consider different forms of measuring the representation (and satisfaction) of voters. 

\begin{table}[t]
    \centering
  \resizebox{\textwidth}{!}{   
  \begin{tabular}{lccccc}
\toprule
{} &   PAV score 		&  JR violations  			&  EJR+	violations	  & priceability violations & priceability  gap \\
\midrule
AV          &     2.502 	&       28.347 &   32.343   &                 124.341 &              6.21 \\
SAV         &      2.547 	&       25.143 &   32.321   &                 116.692 &              7.14 \\
seq-PAV     &      2.585 	&        0 &    0  &                 12.127 &              0.53  \\
Phragmms    &    2.578 		&        0 (Guarantee) &    0  &                 0.000 &             -0.12 \\
seq-Phrag. 	&    2.578 		&        0 (Guarantee) &    0 	&                 0.000 &             -0.16 \\
MES        	&    2.579 		&        0 (Guarantee) &    0 (Guarantee)   &                 0.000 &             -0.17 \\
\bottomrule
\end{tabular}  
}
\caption{Measures related to underrepresentation. Entries marked with a ``(Guarantee)'' are guaranteed to be zero.}
    \label{tab:underrep}
\end{table}

\subsubsection{PAV Score}
We start with the PAV score of the computed committees (see \Cref{tab:underrep}), which is usually regarded as a simple proportionality measure. 
Unsurprisingly, seq-PAV, which greedily optimizes this value in a sequential fashion, performs best. 
The other three proportional rules all have a very similar performance and perform only marginally worse than seq-PAV (i.e., by about $0.2\%$), while still outperforming AV and~SAV. 

\subsubsection{JR and EJR+}
In addition, proportionality is typically judged in terms of binary proportionality axioms, which are often defined in terms of cohesive groups, i.e., subgroups of voters of a certain size that jointly approve some number of candidates. However, because the respective notions tend to be computationally intractable to check, we follow the approach of \citet{BrPe23a} and focus on  notions regarding non-selected candidates instead. 
The intuition here is as follows: if a non-selected candidate is approved by ``many'' voters who are currently ``underrepresented'' in the committee, then the candidate should have been added to the committee. 
Two axioms that implement this general intuition are JR \citep{ABC+16a} and EJR+ \citep{BrPe23a}.\footnote{EJR+ implies JR as well as 
other established proportionality notions such as  
EJR \citep{ABC+16a}, PJR \citep{SFF+17a}, and IPSC \citep{AzLe21a}.}

\begin{definition}
    For a given election $E=(C,V,A,w,k)$ and committee $W\subseteq C$, a non-selected candidate $c\in C\setminus W$ violates 
    EJR+ 
    if there is an $\ell$-supporting group $V'$ of $c$ 
    such that all voters from $V'$ approve less than $\ell$ candidates from~$W$.
    If $c$ violates EJR+ for $\ell=1$, $c$  violates  
    JR. 
\end{definition}
MES is guaranteed to output committees satisfying  EJR+, while committees returned by Phragmms and seq-Phragmén always satisfy JR but may fail EJR+, and the other rules fail even JR \citep{BrPe23a,CeSt21a,LaSk22a}.
However, in our instances the behavior of the rules is quite different: 
All four proportional rules return committees satisfying EJR+ for all tested instances.
In contrast, AV violates JR and EJR+ on all instances, while
SAV violates JR in all but $36$ instances and EJR+ in all but $26$ instances.
To get a more differentiated view, in \Cref{tab:underrep} we present the average number of non-selected candidates violating JR/EJR+ in our elections. 
Given the reported numbers, one can conclude that the committees returned by AV and SAV are typically quite far away from satisfying the two axioms. 
Moreover, as for both rules the number of candidates violating JR and EJR+ are quite similar to each other, it follows that if a candidate violates EJR+, then they often violate JR as well.

\subsubsection{Average Satisfaction $\ell$-Supporting Groups}
EJR+ and JR only check for $\ell$-supporting groups in which \emph{all} voters are unsatisfied with the current assignment. 
Weakening this approach, one can also search for $\ell$-supporting groups with low \emph{average} satisfaction.
This view is reflected in the notion of representativeness recently introduced by \citet{BrPe23a}, which is related to the proportionality degree \citep{Skow21a}. In fact, \citet{BrPe23a} proved that for each rule satisfying EJR+ (such as MES), each $\ell$-supporting group is guaranteed to have an average satisfaction of at least $\frac{\ell-1}{2}$ with the committee.
In each election, we compute for each $\ell\in [k]$, the $\ell$-supporting group that has the lowest average weighed satisfaction. 
In \Cref{fig:representation3}, we show this value for varying $\ell\in [k]$ (averaged over all elections where an $\ell$-supporting group exists). 
On average, all four proportional rules produce outcomes in which all $\ell$-supporting groups have an average satisfaction clearly above $\ell-1$. Thus, they outperform their known worst-case guarantees from the literature, and even consistently outperform the best possible guarantee of $\ell-1$ \citep{BrPe23a}. 
Interestingly, seq-PAV performs slightly better than the other proportional rules. 
SAV and AV perform substantially worse, yet still acceptable for $\ell\ge 4$.

\begin{figure}[t]
\centering
   \resizebox{0.33\textwidth}{!}{
\begin{tikzpicture}[every plot/.append style={line width=2.2pt}]

\definecolor{crimson2143940}{RGB}{214,39,40}
\definecolor{darkgray176}{RGB}{176,176,176}
\definecolor{darkorange25512714}{RGB}{255,127,14}
\definecolor{forestgreen4416044}{RGB}{44,160,44}
\definecolor{mediumpurple148103189}{RGB}{148,103,189}
\definecolor{orchid227119194}{RGB}{227,119,194}
\definecolor{sienna1408675}{RGB}{140,86,75}
\definecolor{steelblue31119180}{RGB}{31,119,180}

\begin{axis}[
legend columns=2, 
legend cell align={left},
legend style={
  fill opacity=0.8,
  draw opacity=1,
  draw=none,
  text opacity=1,
  at={(0.5,1.35)},
  line width=3pt,
  anchor=north,
   /tikz/column 2/.style={
  	column sep=10pt,
  }, font=\LARGE
},
legend entries={AV,
	seq-Phrag.,
	SAV, Phragmms,
	seq-PAV, MES},
tick align=outside,
tick pos=left,
x grid style={darkgray176},
xlabel={$\ell$-supporting groups},
xmin=0.4, xmax=13.6,
xtick style={color=black},
y grid style={darkgray176},
ylabel={minimum average satisfaction},
ymin=-0.749999973283545, ymax=15.7499994389544,
ytick style={color=black},every tick label/.append style={font=\LARGE}, 
label style={font=\LARGE},
xticklabels={$2$,$4$,$6$,$8$,$10$,$12$,$14$},
xtick={2,4,6,8,10,12,14}
]
\addlegendimage{steelblue31119180}
\addlegendimage{mediumpurple148103189}
\addlegendimage{darkorange25512714}
\addlegendimage{crimson2143940}
\addlegendimage{forestgreen4416044}
\addlegendimage{sienna1408675}
\addplot [semithick, steelblue31119180, mark=x, mark size=3, mark options={solid}]
table {%
1 0
2 0
3 1.23065499284617
4 2.14897086737261
};
\addplot [semithick, darkorange25512714, mark=x, mark size=3, mark options={solid}]
table {%
1 0.0113899267737008
2 0.285268768611994
3 0.882560833919876
4 3.290755548152
5 4.35905006635215
6 4.80653877178278
7 7.32698244496995
8 10.6874509404321
};
\addplot [semithick, forestgreen4416044, mark=x, mark size=3, mark options={solid}]
table {%
1 1.09156130282997
2 2.91281020527273
3 4.8239650591232
4 6.01557613392188
5 7.44233581631182
6 9.2885679465521
7 10.7132663751508
8 13.0941963478402
9 13.8895691354364
10 14.1237467106078
11 14.5263152637431
};
\addplot [semithick, crimson2143940, mark=x, mark size=3, mark options={solid}]
table {%
1 0.958833766122003
2 2.06796526734476
3 3.81352195416874
4 5.16941713380989
5 6.84004232038851
6 8.84769193393223
7 10.22817645592
8 12.4992139813201
9 12.6932571445077
10 13.5096763667421
11 13.7217185241465
12 14.301369008083
13 14.9999994656709
};
\addplot [semithick, mediumpurple148103189, mark=x, mark size=3, mark options={solid}]
table {%
1 0.990541724045681
2 2.27377627576019
3 4.48339900662595
4 5.45529017283971
5 6.67802138053512
6 8.18023379834775
7 9.9675174182658
8 12.4111648400482
9 12.9368879205209
10 13.4684085218079
11 13.6995403337828
12 14.3149988691435
13 14.9999994656709
};
\addplot [semithick, sienna1408675, mark=x, mark size=3, mark options={solid}]
table {%
1 0.980355319797568
2 2.32355847663547
3 4.48276043826275
4 5.47659209309329
5 6.7734162804944
6 8.19974835512682
7 10.0320707870978
8 12.346833842223
9 12.8540221867216
10 13.437851877755
11 13.619084026677
12 14.3575119692276
13 14.9999994656709
};
\addplot [line width=1.5pt, black,dashed]
table {%
1 0
2 1
3 2
4 3
5 4
6 5
7 6
8 7
9 8
10 9
11 10
12 11
13 12
};
\end{axis}

\end{tikzpicture}}
  \captionof{figure}{Minimum average satisfaction of $\ell$-supporting groups. The dashed line is the function $f(\ell)=\ell-1$. Lines stop in case no $\ell$-supporting group of this size exists. }
\label{fig:representation3} 
\end{figure}
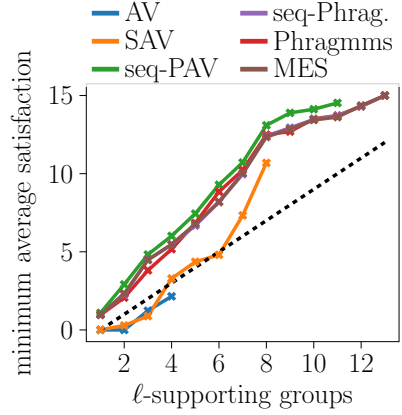

\subsubsection{Priceability}
Complementing the above analysis,  we propose a quantitative measure based on the notion of \textit{priceability} \citep{PeSk20b}, which checks whether voters have an equal influence on the outcome:

\begin{definition}[\citealp{PeSk20b}]
    A committee $W\subseteq C$ is priceable in a given election $E=(C,V,A,w,k)$ if there exists a price system $(p,f)$ with price $p\in [0,1]$ and a cost assignment $f:V\times W \mapsto [0,1]$ such that:
    \begin{enumerate}
        \item $\sum_{c\in W} f(v,c)\leq w(v)$ for all $v\in V$,
        \item $f(v,c)=0$ for all $v\in V$ and $c\in W \setminus A_v$,
        \item $\sum_{v\in V} f(v,c)=p$ for all $c\in W$, and 
        \item $\sum_{v\in V_c} (w(v)-\sum_{c\in W} f(v,c))\leq p$ for all $c\in C \setminus W$. 
    \end{enumerate}
\end{definition}
Seq-Phragmén, Phragmms, and MES always return priceable committees \citep{PeSk20b,CeSt21a}. 
A natural way to turn this axiom into a quantitative measure is to compute the price system that minimizes the \emph{priceability gap}, i.e., the gap between the determined price and the maximum  spare money of the supporters of a non-selected candidate, i.e., 
\[\min_{c\in C\setminus W} \sum_{v\in V_c} (w(v)-\sum_{c\in W} f(v,c))-p.\] 
This captures by how much the candidate is exceeding or missing to be affordable by their supporters. 
Using an LP, we computed the price system minimizing this value.
In the second-to-last and last columns of \Cref{app:rep_app}, we show the average number of candidates whose supporters' money exceeds the set price and the average priceability gap (normalized by the set price), respectively. 
Moreover, in \Cref{fig:representation1} we show the average normalized money by which candidate's supporters exceed the set price  (where we first sort candidates decreasingly by this value and then take the average).

SAV and AV return committees that are far away from being priceable: In all elections, they returned committees violating priceability with an average priceability gap of $6.21$ and $7.14$ and typically around $100$ candidates whose supporters own more than the set price: In fact, in more than $75\%$ of the elections, there are non-selected candidates whose supporter's money exceeds the price by more than $600\%$.
This is a clear indicator that certain voters are much better represented and had a much larger influence on the returned outcome than others in the selected committee. 

For seq-PAV, the picture is mixed: The rule returns priceable outcomes in $327$ of the elections, yet in case the returned outcomes are not priceable, the priceability gap is sometimes as high as $4$, which leads to an average priceability gap of $0.53$ and on average $12$ candidates whose supporters exceed the set price. 

For the other rules, it is guaranteed that the returned committees are priceable.
Accordingly, their priceability gap is negative. 
In fact, there is a small difference between these rules: For MES and seq-Phragmén non-selected candidates are slightly further away from being affordable by their supporters than for Phragmms. 

\section{Preventing Overrepresentation}
\label{sec:over} 

One of the major concerns of blockchain designers is the security of the chain: If a certain fraction of participants collude and together execute some malicious action, they can seize control over the chain and threaten the integrity of the whole system. To protect against such attacks in Nominated Proof-of-Stake, it is vital to ensure that groups of candidates can only get selected if their joint set of supporters has a sufficient stake. In other words, we want to prevent overrepresentation \citep{CeSt21a}.
\footnote{The critical threshold of elected malicious candidates depends on the type of consensus and is $\frac{k}{3}$ in Byzantine fault-tolerant consensus \citep{DBLP:journals/jacm/PeaseSL80}, as used in Polkadot, and $\frac{k}{2}$ in Nakamoto consensus \citep{cryptoeprint:2018/400}. 
However, already a small number of malicious agents might pose certain inconveniences for the system \citep{CeSt21a}.}
In the following, we explore three perspectives on overrepresentation. 

\subsubsection{Minimum Approval Weight}
The first measure we consider was introduced by \citet{CeSt21a}. In order to make it as difficult as possible for an attacker to get $\ell$ committee members selected, they propose to maximize $\min_{W'\subseteq W: |W'|=\ell} w(W')$, i.e., the minimum approval weight of a group of $\ell$ committee members.\footnote{For $\ell=1$, this value is the minimum approval weight of a selected candidate and is maximized by AV (see \Cref{tab:overrep}).} 
\Cref{fig:overrep2} depicts these values for a varying value of~$\ell$.\footnote{
Computing these values is NP-hard, which is why we resorted to an ILP (see Appendix~\ref{app:minimum_support} for details). 
As solving a single instance of the ILP took sometimes more than one day, 
in \Cref{fig:overrep2} we only averaged over $15$ instances, uniformly spaced over our dataset, and considered values of $\ell$ from $\{1,15,30,45,\dots, 300\}$.}
For all examined rules, the values are generally quite high and close to the dashed line (which corresponds to the function $\frac{1}{300}x$). 
This is reassuring, as it means that groups of candidates in the selected committee are also backed by an appropriate amount of stake. 
Considering the differences between the rules, we see that the four proportional rules perform best (for $\ell\geq 15$), with seq-PAV performing slightly worse than the other three. AV performs worst; yet, the generally small difference between AV and the proportional rules is quite remarkable given that we have seen in \Cref{sec:under} that AV tends to underrepresent voter groups (and thus runs the risk of overrepresenting others).

\begin{table}[]
    \centering
    \begin{tabular}{llllllr}
\toprule
{} 			   &   \multirowcell{2}[0pt][l]{maximin\\support\\value} &	\multirowcell{2}[0pt][l]{min. appr.\\weight of\\winner} 	& \multicolumn{3}{c}{cost of ``replacing'' $\ell$} \\
\cmidrule{4-6   }
{} 			      &   		   		&  		&  	$\ell=1$ & $\ell=\frac{k}3$ & $\ell=\frac{k}2$ \\
\addlinespace[0.5em] \midrule
AV          &           0.0015 &         0.0110                &        0.0110 &0.14 &       0.27   \\
SAV         &           0.0018 &         0.0024                &       0.0015 &  0.24&       0.40              \\
seq-PAV      &           0.0024 &         0.0028                &     ? & ? &     ?  \\
Phragmms    &           0.00272 &         0.0028                &       0.0027 & 0.40 &      0.79   \\
seq-Phrag. &           0.00270 &         0.0028                &       0.0027 &  0.40  &    0.79  \\
MES         &           0.00269 &         0.0028                &           ? &           ? &                ?  \\
\bottomrule
\end{tabular}
    \caption{Measures related to overrepresentation.} 
\label{tab:overrep}
\end{table}

\subsubsection{Maxminim Support Value}
As an aggregate version of this measure, \citet{CeSt21a} also proposed to consider the minimum \textit{average} approval weight of a committee member, where the minimum is taken over groups of different sizes, i.e., 
$\min_{\substack{W'\subseteq W}} \frac{1}{|W'|}w(W')$. Interestingly, \citet{CeSt21a} proved that this value is equivalent to the \textit{maximin support (MMS) value}, which was introduced in a different context (see \Cref{def:mms} below). 
 In \Cref{tab:overrep}, we see that, on average, the four proportional rules achieve  substantially higher MMS values than AV and SAV. 
In particular, for the four proportional rules, the MMS value is very close to the minimum approval weight of a selected candidate (also in \Cref{tab:overrep}), which constitutes a natural upper bound for it. 
Taking a closer look at the four proportional rules, seq-PAV performs slightly worse, while  Phragmms, seq-Phragmén, and MES all produce very similar values. 
Particularly interesting is the comparison between seq-Phragmén and Phragmms: 
The main argument in favor of a potential switch from the former to the latter is the fact that the latter provides a constant-factor approximation guarantee for the MMS value \citep{CeSt21a}. Among the $252$ instances in our dataset where seq-Phragmén and Phragmms select different committees, Phragmms outperforms seq-Phragmén in $209$ cases. (In the remaining $43$ instances, seq-Phragmén wins.) However, the difference between the two rules is always at most $0.0002$ and thus negligible.

\subsubsection{Stake Lost}
For our next measure, we take a closer look at how slashing works in Polkadot. As explained in \Cref{sec:polkadot}, when committee members misbehave, some of their supporters lose stake. To determine the amount of stake a voter loses in this event, the election mechanism in Polkadot not only outputs a committee $W$, but also a vote assignment $\alpha$, 
where $\alpha(v,c)$ specifies how much of a voter $v$'s stake is assigned to committee member $c \in A_v$ (see \Cref{sec:abc}).
Following an approach proposed by \citet{SFFB22b}, this vote assignment is chosen so as to maximize the backing weight 
of the least-backed committee member.

\begin{definition}\label{def:mms}
Given an election $E=(C,V,A,w,k)$ and a committee~$W$, the \emph{maximin support value} of $W$ is given by
$\max_{\alpha} \min_{c\in W} \sum_{v\in V_c} \alpha(v,c)$, 
where the maximum is taken over all possible vote assignments for $W$. A vote assignment maximizing this quantity is called a \emph{maximin assignment}.
\end{definition}

If a committee member $c$ misbehaves in Polkadot, each supporter $v \in V_c$ loses (up to) $\alpha(v,c)$ stake \citep{polkadot-overview}. 
Accordingly, the maximin support value expresses the minimum amount of stake that is slashed if a single member of the committee acts maliciously.
For a given committee, a maximin assignment can be computed via an LP \citep{SFFB22b}. Using this assignment, in \Cref{fig:overrep4} we plot the minimum amount of stake assigned to a group of $\ell$ committee members, for $\ell \in [k]$. This corresponds to the minimum stake that will be slashed if $\ell$ committee members act maliciously. 

\Cref{fig:overrep4} is closely connected to \Cref{fig:overrep2}; the main difference is that in \Cref{fig:overrep4} we assume that there is a predefined split of the stake of supporters onto their approved candidates (as defined by the maximin assignment), whereas in \Cref{fig:overrep2} we consider the full stake of all voters approving at least one candidate from the group. 
Accordingly, in \Cref{fig:overrep4}, the dashed line is an upper bound  (that is achieved if all committee members have the same backing weight). 
The difference between the rules is much more pronounced in \Cref{fig:overrep4}, with the proportional rules performing much better than AV and SAV. 
In particular, for the critical thresholds of $k/3$ and $k/2$, the difference is above $50\%$. 
The proportional rules, in turn, are reasonably close to the ideal line.

\begin{figure*}[t]
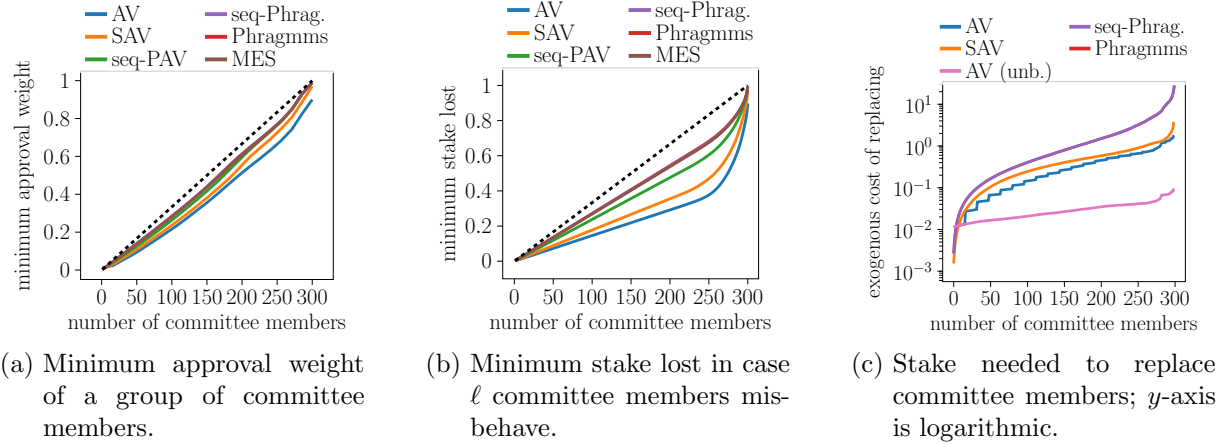

    \centering
\begin{subfigure}[b]{.3\textwidth}
  \centering
  \resizebox{0.95\textwidth}{!}{
\begin{tikzpicture}[every plot/.append style={line width=2.2pt}]

\definecolor{crimson2143940}{RGB}{214,39,40}
\definecolor{darkgray176}{RGB}{176,176,176}
\definecolor{darkorange25512714}{RGB}{255,127,14}
\definecolor{forestgreen4416044}{RGB}{44,160,44}
\definecolor{mediumpurple148103189}{RGB}{148,103,189}
\definecolor{sienna1408675}{RGB}{140,86,75}
\definecolor{steelblue31119180}{RGB}{31,119,180}

\begin{axis}[
legend columns=2, 
legend cell align={left},
legend style={
  fill opacity=0.8,
  draw opacity=1,
  draw=none,
  text opacity=1,
  at={(0.5,1.35)},
  line width=3pt,
  anchor=north,
   /tikz/column 2/.style={
  	column sep=10pt,
  }, font=\LARGE
},
legend entries={AV,
	seq-Phrag.,
	SAV, Phragmms,
	seq-PAV, MES},
tick align=outside,
tick pos=left,
x grid style={darkgray176},
xlabel={number of committee members},
xmin=-30, xmax=330,
ymin=-0.05,ymax=1.05,
xtick style={color=black},
y grid style={darkgray176},
ylabel={minimum approval weight},
ytick style={color=black},every tick label/.append style={font=\LARGE}, 
label style={font=\LARGE},
legend image post style={line width =3pt}
]
\addlegendimage{steelblue31119180}
\addlegendimage{mediumpurple148103189}
\addlegendimage{darkorange25512714}
\addlegendimage{crimson2143940}
\addlegendimage{forestgreen4416044}
\addlegendimage{sienna1408675}
\addplot [semithick, steelblue31119180]
table {%
	1	0.012277353170597
	16	0.023266577614533
	31	0.054286077986586
	46	0.084675087754815
	61	0.120022519027384
	76	0.156486869529888
	91	0.193701510145222
	106	0.23156619064231
	121	0.272427612957695
	136	0.315617562520017
	151	0.358625203775042
	166	0.405145280746251
	181	0.453677903894041
	196	0.501449970234367
	211	0.546442335906129
	226	0.589863304921715
	241	0.634979840731122
	256	0.684124150149424
	271	0.743138867639693
	286	0.827442536942583
 300 0.9
	
};
\addplot [semithick, darkorange25512714]
table {%
1	0.002612506862115
16	0.030998329298331
31	0.064022417905725
46	0.098858234377857
61	0.135745829920589
76	0.173298294995817
91	0.212158807922918
106	0.253029884470554
121	0.294944753648301
136	0.338401282778109
151	0.384258004749105
166	0.431863162518517
181	0.481356206496186
196	0.538142251275657
211	0.597046168760538
226	0.64554078272196
241	0.692919171656262
256	0.746627941592946
271	0.806363906823633
286	0.889589557419398
300 0.97
};
\addplot [semithick, forestgreen4416044]
table {%
1	0.002743052984473
16	0.039519302651988
31	0.07674465687142
46	0.114980839058059
61	0.154699382225555
76	0.195438007469032
91	0.237473154691566
106	0.28086367724042
121	0.326053540200754
136	0.372028852192664
151	0.419962309312388
166	0.470545940752129
181	0.524878182847724
196	0.580636026997596
211	0.635405715242584
226	0.685331984834575
241	0.733699494771056
256	0.784424477850052
271	0.846654862998143
286	0.932427181593293
300 0.99
};
\addplot [semithick, crimson2143940]
table {%
	1	0.002802433425633
	16	0.044187838273836
	31	0.086071124878742
	46	0.128102196980034
	61	0.170849688158115
	76	0.214062715516484
	91	0.257732692099198
	106	0.301982047302902
	121	0.347229718095303
	136	0.393397514413243
	151	0.440929273671944
	166	0.490877415029155
	181	0.544186661201183
	196	0.595509677525469
	211	0.64295896850596
	226	0.688971343054911
	241	0.736176987717759
	256	0.787832296822109
	271	0.848517810746072
	286	0.930761498988192
 300 0.99
};
\addplot [semithick, mediumpurple148103189]
table {%
1	0.002784235774695
16	0.043693377247923
31	0.085350152894
46	0.127315988913862
61	0.169947046628775
76	0.213241809623274
91	0.256945103298775
106	0.301613538914171
121	0.347115494312106
136	0.394439027387915
151	0.44422668241785
166	0.497229500861161
181	0.5497398399197
196	0.601607007599228
211	0.647360906782731
226	0.692745395938793
241	0.739435492170348
256	0.789556828335065
271	0.849683427399865
286	0.932417756144607
300 0.99
};
\addplot [semithick, sienna1408675]
table {%
1	0.00278922857873
16	0.043502117841345
31	0.084938005711393
46	0.126830538472328
61	0.169173860262643
76	0.212216925092118
91	0.25579866243708
106	0.300076978141419
121	0.34527075953128
136	0.392443991140671
151	0.441362386862774
166	0.49353622091995
181	0.5479890772358
196	0.599764392873554
211	0.647598711965972
226	0.692716619549719
241	0.73912526471513
256	0.789004800742462
271	0.848969253635834
286	0.931679425434244
300 0.99
};
\addplot [dashed, black]
table {%
1	0.00333333333
16	0.05333333333
31	0.10333333333
46	0.15333333333
61	0.20333333333
76	0.25333333333
91	0.30333333333
106	0.35333333333
121	0.40333333333
136	0.45333333333
151	0.50333333333
166	0.55333333333
181	0.60333333333
196	0.65333333333
211	0.70333333333
226	0.75333333333
241	0.80333333333
256	0.85333333333
271	0.90333333333
286	0.95333333333
300	1	
};
\end{axis}

\end{tikzpicture}}
  \caption{Minimum approval weight of a group of committee
members.}
  \label{fig:overrep2}

\end{subfigure}\hfill
\begin{subfigure}[b]{.3\textwidth}
  \centering
  \resizebox{0.95\textwidth}{!}{\input{./plots_aaai/lost_stake_maximin.tex}}
  \caption{Minimum stake lost in case $\ell$ committee members misbehave.}
  \label{fig:overrep4}
\end{subfigure}\hfill
\begin{subfigure}[b]{.3\textwidth}
  \centering
  \resizebox{0.95\textwidth}{!}{\input{./plots_aaai/cost_of_buying.tex}}
  \caption{Stake needed to replace committee members; $y$-axis is logarithmic.}
  \label{fig:overrep3}
\end{subfigure}
\caption{Different figures for metrics to prevent overrepresentation.  The dashed line is the function $\frac{1}{300}\ell$.}
\label{fig:overrep}
\end{figure*}

\subsubsection{Exogenous Cost of Replacing}
For our final measure, we take an ``exogenous'' view and reason about how much stake a malicious agent would need to possess in order to \textit{replace} a given number of committee members by newly added candidates (assuming that the agent is allowed to add new votes and candidates, while the remainder of the election remains unchanged). 
To the best of our knowledge, this view has not been explored so far. 
We show that, for most of our rules, the cost of replacing $\ell$ candidates can be computed in constant time, provided we have access to data that is generated  while computing the rules. 

\begin{theorem}[informal]\label{theorem1}
    For AV, SAV, seq-Phragmén, and Phragmms, the minimum exogenous cost of replacing $\ell$ candidates can be computed in $\mathcal{O}(1)$ time, assuming we use data from executions of the rules (specified in Appendix~\ref{app:replacing}).
\end{theorem}

For details, we refer to Appendix~\ref{app:replacing}, where we also discuss the complexity of this problem for seq-PAV and MES. 
\Cref{fig:overrep3} shows the stake needed to replace $\ell$ committee members, for $\ell \in [k]$ (the $y$-axis is logarithmic and the pink line stands for AV where approval votes of unbounded length are allowed). 
AV has the highest cost for $\ell=1$, since AV maximizes the minimum approval weight of a selected candidate (see \Cref{tab:overrep}).
For larger $\ell$, replacements under AV become cheaper, as the same stake can be used to replace multiple candidates. 
In contrast, for the other rules (except seq-PAV and MES), there always exists an optimal replacement in which every voter approves one candidate. 
However, already for $\ell=4$ the cost of replacing starts to be more expensive for seq-Phragmén and Phragmms than for AV. 
Notably, external attacks for AV would be even cheaper if voters were allowed to approve an unbounded number of candidates (see the pink line in \Cref{fig:overrep3}). 
For larger $\ell$, the cost for seq-Phragmén and Phragmms becomes substantially higher than for AV and even SAV (see also \Cref{tab:overrep}).
In particular, for both seq-Phragmén and Phragmms, replacing one-third of the committee would require $40\%$ of the total stake of all voters. 
This is roughly the total stake possessed by agents not participating in the validator elections, and thus makes such an external attack highly unlikely.

\section{Analyzing Design Decisions}\label{sec:desdec}
In this section, we briefly summarize our findings regarding the influence of various design decisions on our measures. For the full analysis and relevant data, see Appendix~\ref{app:desdec}.

\subsection{Choosing the Committee Size}
In Polkadot, the committee size is determined by the network's governance body of token holders and is thus an adjustable design choice. 
This observation raises the following general question: Based on the desiderata formulated in the previous sections, would it be beneficial to increase or decrease the committee size?

We focus on $k=200$ and $k=400$ as two alternative committee sizes (Appendix~\ref{app:desdec} also includes $k=250$ and $k=350$). Generally speaking, for $k=200$ and $k=400$, the differences between the rules are similar as for $k=300$, so we only focus on the general trends in the results. 
Regarding underrepresentation, increasing $k$ turns out to be clearly favorable: For $k=400$, the minimum average satisfaction of $\ell$-supporting groups increases substantially.
Regarding overrepresentation, it turns out that the committee size has only a marginal influence on attacks trying to take over one-third or half of the committee:
Intuitively speaking,  individual candidates in a committee with more candidates are less backed; 
on the other hand, one-third/half of the committee corresponds to larger numbers of candidates. 
It turns out that both effects approximately cancel out for the cost of replacing. 
By contrast, the minimum approval weight values slightly drop by around $15\%$ when increasing $k$ to~$400$.
 
\subsection{Selecting Candidates Multiple Times}
\citet{BGP+19a} initiated the study of ABC elections where 
multiple copies of a candidate can be included in the committee. 
Here, we analyze what impact this would have on our measures.\footnote{While  this is currently not possible in the Polkadot network, it is very easy for candidates to create copies of themselves. Indeed, entities that own a large amount of stake create multiple (up to $100$) candidates in the system and even label them as copies. 
Thus, \textit{de facto} the same entity can already control multiple candidates.}

Generally speaking, allowing copies in the committee is favorable. Most notably, it leads to a more uniform distribution of candidates' backing weights in the maximin assignment, leading to a drop of their 
variance 
by around $80\%$. 
The only metrics where allowing copies leads to a worse performance regard the satisfaction of $\ell$-supporting groups. This is to be expected, as allowing for copies implies that \textit{all} candidates count as ``non-selected'' and, thus, our measures range over strictly more subgroups of voters. 
Nevertheless, the minimum average satisfaction of $\ell$-supporting groups remains quite high, suggesting that allowing for copies might be a worthwhile consideration for platform designers. 
Remarkably, whether or not to allow copies has a stronger (positive) impact on several measures than changing the committee size to $200$ or $400$ (or choosing among the proportional rules). In particular, this holds for the PAV score, the cost of replacing candidates, and the minimum approval weight of groups of committee members.

\section{Discussion and Conclusion}\label{sec:conc}
We conducted a thorough multi-criteria analysis of the behavior of different ABC voting rules in Polkadot elections. We conclude with a discussion of our data and results followed by a summary of different directions for future work.

\subsubsection*{Discussion and Summary of Results}

Generally speaking, we found that the ``proportional'' rules (seq-PAV, seq-Phragm\'en, MES, and Phragmms) behave very similarly. 
Only seq-PAV showed a slightly different behavior: Our results suggest that the rule prioritizes the total weighted satisfaction of voters slightly more, even if this means that some voters have a greater influence on the outcome than others.
One reassuring observation for Polkadot developers is that all four rules provide a good level of security.  
For instance, for both seq-Phragm\'en and Phragmms, seizing control of the network by replacing one-third of the committee with malicious candidates requires around $40\%$ of the tokens present in the election. 

In contrast, SAV and the common AV rule tend to return committees that are less similar to the ones produced by the proportional rules. In particular, they frequently violate basic representation axioms and return outcomes overrepresenting certain voter groups. 
The most drastic difference is that seizing control over one-third of the committee only requires $14\%$ of tokens for AV and $24\%$ for SAV. Nevertheless, the performance of both rules could still be viewed as acceptable with regard to the other security measures. 

We have also explored the impact of various design decisions. Here, we observed that changing the committee size only marginally influences our measures, whereas allowing for selecting (copies of) candidates multiple times is a more impactful and generally beneficial design decision.  

\subsubsection*{Collected Data}
We have collected $496$ and\ $\num{1520}$ large approval-based committee elections from the Polkadot and Kusama  network, respectively and hope that others will 
use our data which is available on \url{github.com/n-boehmer/ABC-practice-Polkadot}. 
One might argue that the fact that all voting rules substantially outperform their axiomatic guarantees on our collected elections (e.g., seq-PAV always fulfills EJR+) suggests that voters' preferences must be quite ``simple''.
While measuring the ``simplicity'' of an election is non-trivial, we calculated the following quantity as a proxy: For any two voters with overlapping approval set, we divide the size of the intersection of their approval sets by the sum of the sizes of their approval sets. If we would be in the so-called ``party-list'' setting, i.e., any pair of approval sets are disjoint or equal, then this value would be $0.5$. However, we found that the average of this value is typically around $0.1$, proving that our elections are far from falling into the simple party-list setting.
Instead, our interpretation of the observed phenomenon is rather that most of our considered voting rules, e.g., seq-PAV, satisfy demanding proportionality axioms on most (reasonable) elections, only violating them on adversarial constructed ones unlikely to occur in real-world or synthetic data. Indeed, this view is supported by the works of \citet{mapofRules}, \citet{BFNK19a}, and \citet{BrPe23a}, who also report that voting rules substantially outperform theoretical guarantees and that even random committees often satisfy strong proportionality notions. 

\subsubsection*{Future Work}
Several lines of continuation of this work can be envisioned. First,
it would be interesting to dive deeper into the analysis of different voting rules, trying to identify properties of ``controversial'' candidates (i.e., candidates that are only selected by some rules but not others).

Second, it would be intriguing to analyze the connection between over- and underrepresentation in more detail, both from a theoretical and an empirical perspective; for instance, it would be interesting to relate under- and overrepresentation axioms to each other formally, or to give bounds on the ``gap'' between a rule's over- and underrepresentation performance. 
Moreover, the goal of preventing overrepresentation (motivated by Polkadot's security concerns)  by itself constitutes a new research direction deserving further attention and intense study.

Third, 
our finding that classic ``binary'' axioms are not able to meaningfully distinguish voting rules  in practice highlights the necessity for further theoretical work on axioms. 
Our hope is that our work initiates a rethinking of the standard binary axiomatic analysis and a shift more into a direction of \textit{quantitative} axioms (e.g. the \emph{cost of replacing} candidates). Our paper takes a first step in this direction, but given the multitude of binary axiomatic works, there are clearly numerous directions remaining.

\section*{Acknowledgments}
This work was initiated at Dagstuhl Seminar 22271 ``Algorithms for Participatory Democracy'' in July 2022 (\url{https://www.dagstuhl.de/22271}).
We thank Jannik Peters for helpful comments. 
Markus Brill was supported by the Deutsche Forschungsgemeinschaft (DFG) under grant BR 4744/2-1. 
Luis S\'anchez-Fern\'andez was supported by the “Generation of Reliable Synthetic Health Data for Federated Learning in Secure Data Spaces” Research Project (PID2022-141045OB-C43 (AEI/ERDF, EU)) funded by MCIN/AEI/ 10.13039/501100011033 and by “ERDF A way of making Europe” by the “European Union”. Ulrike Schmidt-Kraepelin was supported by the \emph{Centro de Modelamiento Matemático (CMM)} (under grant FB210005, BASAL funds for center of excellence from ANID-Chile), \emph{ANID-Chile} (grant ACT210005) as well as by the National Science Foundation under Grant No. DMS-1928930 and the Alfred P. Sloan Foundation under grant G-2021-16778 while she was in residence at the Simons Laufer Mathematical Sciences Institute (formerly MSRI) in Berkeley, California, during the Fall 2023 semester. The main work was done while Niclas Boehmer was affiliated with TU Berlin, where he was supported by the DFG project ComSoc-MPMS (NI 369/22).

\section*{Ethics Statement}
This paper acknowledges the significant environmental impact associated with blockchain technology.
However, the Polkadot network employs a Proof-of-Stake (PoS) consensus mechanism that does not involve energy-intensive mining activities that are common to the more well-known Proof-of-Work (PoW) systems. 
This results in an energy usage that is several orders of magnitude smaller. 
Polkadot has been acknowledged as one of the protocols with the lowest carbon footprint having an annual CO2 emission equivalent to that of five average American households \citep{report}.
This is partly because the number of validators who participate in consensus is limited enabled by the election process studied in the paper. 
 
\bibliographystyle{plainnat}

\clearpage 
\appendix

\section*{Appendix}

\section{Additional Material for \Cref{sec:under}}

\begin{table*}[!htb]
    \centering
  \begin{tabular}{lcccccccccc}
\toprule
{} &     weighted  &  max. approval weight& \multicolumn{3}{c}{min. average satisfaction}  \\
{} &      satisfaction  & loser&    $\ell=1$ &  $\ell=5$ &  $\ell=10$ \\
\midrule
AV          &                  9.055 &      0.011 &          0.000 &              $-$ &               $-$  \\
SAV         &                  8.724 &      0.022  &          0.011 &            4.359 &               $-$  \\
seq-PAV     &                  8.672 &       0.028 &         1.092 &            7.442 &            14.124 \\
Phragmms      &                  8.650 &        0.039 &       0.960 &            6.840 &            13.510  \\
seq-Phrag.      &                  8.588 &      0.039  &      0.990 &            6.678 &            13.468  \\
MES             &                  8.584 &       0.039 &       0.980 &            6.773 &            13.438  \\
\bottomrule
\end{tabular}  
    \caption{Average values for different notions related to voter satisfaction and  preventing underrepresentation.} \label{app:rep_app}
\end{table*}

\begin{figure*}[!htb]
    \centering
    \begin{minipage}[t]{0.45\textwidth}
  \centering
  \resizebox{0.8\textwidth}{!}{
\begin{tikzpicture}[every plot/.append style={line width=2.2pt}]

\definecolor{crimson2143940}{RGB}{214,39,40}
\definecolor{darkgray176}{RGB}{176,176,176}
\definecolor{darkorange25512714}{RGB}{255,127,14}
\definecolor{forestgreen4416044}{RGB}{44,160,44}
\definecolor{mediumpurple148103189}{RGB}{148,103,189}
\definecolor{sienna1408675}{RGB}{140,86,75}
\definecolor{steelblue31119180}{RGB}{31,119,180}

\begin{axis}[
legend columns=2, 
legend cell align={left},
legend style={
  fill opacity=0.8,
  draw opacity=1,
  draw=none,
  text opacity=1,
  at={(0.5,1.35)},
  line width=3pt,
  anchor=north,
   /tikz/column 2/.style={
  	column sep=10pt,
  }, font=\LARGE
},
legend entries={AV,
	seq-Phrag.,
	SAV, Phragmms,
	seq-PAV, MES},
tick align=outside,
tick pos=left,
x grid style={darkgray176},
xlabel={$\ell$},
xmin=0.35, xmax=14.65,
xtick style={color=black},
y grid style={darkgray176},
ylabel={\#candidates with $\ell$-supporting groups},
ymin=-5.86290322580645, ymax=123.120967741935,
ytick style={color=black},every tick label/.append style={font=\LARGE}, 
label style={font=\LARGE},
xticklabels={$2$,$4$,$6$,$8$,$10$,$12$,$14$},
xtick={2,4,6,8,10,12,14}
]
\addlegendimage{steelblue31119180}
\addlegendimage{mediumpurple148103189}
\addlegendimage{darkorange25512714}
\addlegendimage{crimson2143940}
\addlegendimage{forestgreen4416044}
\addlegendimage{sienna1408675}
\addplot [semithick, steelblue31119180, mark=x, mark size=3, mark options={solid}]
table {%
1 109.10685483871
2 49.2076612903226
3 14.9254032258065
4 2.3366935483871
};
\addplot [semithick, darkorange25512714, mark=x, mark size=3, mark options={solid}]
table {%
1 116.625
2 65.179435483871
3 42.9677419354839
4 25.7258064516129
5 15.0120967741935
6 6.67338709677419
7 3.2258064516129
8 0.171370967741935
};
\addplot [semithick, forestgreen4416044, mark=x, mark size=3, mark options={solid}]
table {%
1 117.258064516129
2 66.491935483871
3 45.1229838709677
4 30.1895161290323
5 17.5604838709677
6 8.73991935483871
7 5.34274193548387
8 1.05443548387097
9 0.491935483870968
10 0.266129032258065
11 0.0766129032258065
};
\addplot [semithick, crimson2143940, mark=x, mark size=3, mark options={solid}]
table {%
1 116.911290322581
2 65.3649193548387
3 42.3709677419355
4 26.3528225806452
5 15.6612903225806
6 10.9274193548387
7 8.40725806451613
8 3.98387096774194
9 3.48991935483871
10 2.61693548387097
11 2.07862903225806
12 0.743951612903226
13 0.0705645161290323
};
\addplot [semithick, mediumpurple148103189, mark=x, mark size=3, mark options={solid}]
table {%
1 117.040322580645
2 65.8487903225806
3 45.6693548387097
4 31.8770161290323
5 20.9576612903226
6 13.1088709677419
7 9.18346774193548
8 4.36290322580645
9 3.70564516129032
10 2.6875
11 2.06653225806452
12 0.67741935483871
13 0.0665322580645161
};
\addplot [semithick, sienna1408675, mark=x, mark size=3, mark options={solid}]
table {%
1 117.066532258065
2 66.2439516129032
3 46.3528225806452
4 32.3608870967742
5 20.9193548387097
6 13.0967741935484
7 8.99395161290323
8 4.33266129032258
9 3.72379032258065
10 2.63104838709677
11 2.04637096774194
12 0.647177419354839
13 0.0584677419354839
};
\end{axis}

\end{tikzpicture}}
  \caption{Number of non-elected candidates with $\ell$-supporting groups. Only non-zero values are plotted.}
  \label{fig:representation2}
\end{minipage}\hfill
\begin{minipage}[t]{0.45\textwidth}
\centering
   \resizebox{0.8\textwidth}{!}{\input{./plots_aaai/prica_gap.tex}}
  \caption{Average amount by which supporter's money exceeds price for priceability system.}
  \label{fig:representation1}
\end{minipage}
\end{figure*}

In this section, we provide some more information regarding the satisfaction and representation of voters (see \Cref{app:rep_app} for an overview).  
We start by considering the average weighted satisfaction of the voters; see the first column of \Cref{app:rep_app}. AV maximizes this objective; the other rules perform between $4\%$ and  $5.3\%$ worse, with SAV producing the best and MES producing the worst result. In general, the observed weighted satisfaction values are remarkably high given that voters approve less than $10$ candidates on average.

To give some more context on the minimum average satisfaction of $\ell$-supporting groups (see \Cref{app:rep_app}), \Cref{fig:representation2} depicts the average number of non-selected candidates for which there is an $\ell$-supporting group in our elections. 
In \Cref{fig:representation2}, we observe that for all rules, the number of non-selected candidates for which an $\ell$-supporting group exists quickly decreases with increasing $\ell$. 
We see a contrast between AV and the other rules, especially with respect to the size of the largest supporting group.
For AV, there are no $\ell$-supporting groups for $\ell>4$, implying that having an approval weight larger than $\frac{5}{300}$ 
is always sufficient to be included in the committee. For the other rules, this is not the case, as here there also exist $\ell$-supporting groups for more candidates and for larger values of $\ell$. While generally speaking all other rules show a similar performance, SAV and seq-PAV tend to produce slightly less non-selected candidates with $\ell$-supporting groups.
This behavior is also reflected in the maximum approval weight of a non-selected candidate shown in column two of \Cref{app:rep_app}.

\section{Additional Material for \Cref{sec:over}}

\subsection{Computing Candidate Groups with Minimum Approval Weight}\label{app:minimum_support}
To generate \Cref{fig:overrep2}, we needed to compute size-$\ell$ candidate sets with minimum approval weight. 
We prove that this problem is NP-hard and give an ILP formulation, which we used in our experiments. 

\begin{proposition}
    Given an election $(C,V,A,w,k)$, a committee $W$, and an integer $\ell\in [|W|]$, computing the size-$\ell$ subset $W'\subseteq W$ with minimum approval weight, i.e., $\argmin_{W'\subseteq W: |W'|=\ell}w(W')$, is NP-hard.
\end{proposition}

\begin{proof}
    We reduce from \textsc{Clique} on regular graphs. 
    Given an $r$-regular graph $G=(U,E)$ and an integer $t$, we construct our election as follows. 
    We add a candidate $c_u$ for each $u\in U$ and a voter for each edge $e\in E$ approving the two candidates corresponding to the vertices incident to $e$. 
    All voters have weight $1$, we set $W$ to be the set of all candidates and ask whether there is a size-$t$ set $W'\subseteq W$ of candidates with $w(W')\leq r\cdot t - {t\choose 2}$. 

    Assume that $U'\subseteq U$ is a size-$t$ clique in $G$, then $W':=\{c_u: u\in U'\}$ fulfills $w(W')\leq r\cdot t - {t\choose 2}$, as there are ${t\choose 2}$ voters that approve two candidates from $W'$. 

    Assume that we are given a size-$t$ candidate set $W'$ with $w(W')\leq r\cdot t - {t\choose 2}$, then $U':=\{u: c_u\in W'\}$ is a size-$t$ clique in $G$. This follows as $w(W')\leq r\cdot t - {t\choose 2}$ implies that there must be ${t\choose 2}$ voters approving two candidates from $W'$, implying that there are  ${t\choose 2}$ edges with two endpoints in $U'$.
\end{proof}

To solve the problem in our experiments, we model it as an Integer Linear Program (ILP) and solve it using Gurobi.
We add a binary variable $x_c$ for each $c\in C$ which is $1$ if $c$ is selected as part of the set and $0$ otherwise. 
We impose that a size-$\ell$ candidate set is selected by requiring that:
$$\sum_{c\in C} x_c=\ell.$$
Further, we add a real variable $y_v$ for each $v\in V$ which is intended to be $0$ if $v$ approves of any of the selected candidates and $1$ otherwise.
We enforce these upper bounds of $0$ or $1$ by adding for each voter $v\in V$ and each candidate $c\in A_v$ the following constraint: 
$$x_c+y_v\leq 1,$$
while lower-bound constraints are not needed. Finally, we add the following objective function to enforce that a subset of minimum approval weight is selected:
$$\min \sum_{v\in V} (1-y_v)\cdot w(v).$$

\subsection{Proof of \Cref{theorem1}} \label{app:replacing}

We now formally define the minimum exogenous cost of an $\ell$-replacement, which was informally introduced in \Cref{sec:over}. 
Consider some weighted multiwinner election $E = (C,V,A,w,k)$. 
We call any $(C',V',A',w')$ with $|C'|=\ell$ and $C \cap C', V \cap V' = \emptyset$ an $\ell$-replacement. 
Its weight is $w(V')=\sum_{v \in V'} w'(v)$, and we say it is \emph{successful} for some voting rule if $C' \subseteq W$, where $W$ is the committee selected by the rule on the 
extended election $(C \cup C',V \cup V', A \cup A', w \cup w',k)$, where $w \cup w'$ denotes the natural concatenation of the two weight functions $w$ and $w'$. 
The \emph{minimum exogenous cost of an $\ell$-replacement} is the minimum weight of a successful $\ell$-replacement. 

In order to compute a minimum $\ell$-replacement for all $\ell \in [300]$ and all of our tested elections, which we depicted in \Cref{fig:overrep3}, we derived algorithms for AV, SAV, seq-Phragmén, and Phragmms. When reporting the running time for any of these algorithms, it is important to point out that we assume that certain information which is naturally derived from the execution of a voting rule on the original instance, can be taken as input for our computations. Since this information may vary for each voting rule, we specify the required information in each theorem statement. 

Since all considered rules are sequential, we can order the selected candidates chronologically: for a given voting rule, 
let $c_1$ be the first selected candidate, let $c_2$ be the second selected candidate, and so on. 
Note that this labeling depends on the rule under consideration. 

\begin{theorem}
    For every $\ell \in [k]$, the minimum exogenous cost of an $\ell$-replacement for AV can be computed in time $\mathcal{O}(1)$, assuming that we get as input the approval weight of the candidate $c = c_{k-\ell +1}$, i.e., $w(V_c)$. 
\end{theorem}

\begin{proof}
    Let $(C',V',A',w')$ be an $\ell$-replacement for election $(C,V,A,w,k)$. It is successful for AV if and only if $$w'(V'_{c'}) \geq w(V_{c_{k-\ell+1}})$$ holds for all $c' \in C'$.%
    \footnote{We assume in this section that any election rule breaks ties in a way that benefits the candidates in the $\ell$-replacement.} 
    Thus, an $\ell$-replacement that consists, for example, of one voter with weight $w(V_{c_{k-\ell+1}})$ approving all $\ell$ candidates in $C'$, is clearly successful and weight minimal.
\end{proof}

\begin{theorem}
    For every $\ell \in [k]$, the minimum exogenous cost of an $\ell$-replacement for SAV can be computed in time $\mathcal{O}(1)$, assuming that we get as input the SAV score of the candidate $c = c_{k-\ell + 1}$, i.e., $x:=\sum_{v \in V_{c_{k-\ell+1}}} \frac{w(v)}{|A_v|}$.
\end{theorem}

\begin{proof}
    Let $(C',V',A',w')$ be an $\ell$-replacement for election $(C,V,A,w,k)$. It is successful for SAV if and only if 
    $$\sum_{v \in V'_{c'}}\frac{w'(v)}{|A_v|} \geq \sum_{v \in V_{c_{k-\ell+1}}} \frac{w(v)}{|A_v|} = x$$ holds for all $c' \in C'$. 
    Thus, an $\ell$-replacement that consists, for example, of $\ell$ voters each with weight $x$ and approving a distinct candidate in $C'$, 
    is clearly successful and weight minimal, with weight $\ell \cdot x$.
\end{proof}

\begin{theorem}
        For every $\ell \in [k]$, the minimum exogenous cost of an $\ell$-replacement for seq-Phragmén can be computed in time $\mathcal{O}(1)$, 
        assuming that we get as input the point in time $t$ at which candidate $c_{k-\ell+1}$ is elected. 
\end{theorem}

\begin{proof}
    Let $t \in \mathbb{R}$ be the point in time at which candidate $c_{k - \ell +1}$ is elected. 
    A necessary condition for an $\ell$-replacement $(C',V',A',w')$ to be successful is that voters in $V'$ earned at least $\ell$ units of money until time $t$, 
    which is the case if and only if $\sum_{v \in V'}w'(v)\geq \frac{\ell}{t}$. 
    To see that this is indeed the minimum exogenous cost of an $\ell$-replacement, consider an $\ell$-replacement that consists of $\ell$ voters of weight $\frac{1}{t}$ each, each approving exactly one (distinct) candidate in $C'$. 
    In the extended election, candidates $c_1,\dots,c_{k-\ell}$ are elected with unchanged order, and then all candidates in $C'$ will be chosen before candidate $c_{k-\ell+1}$, 
    hence the $\ell$-replacement is successful. 
\end{proof}

\begin{theorem}
    For every $\ell \in [k]$, the minimum exogenous cost of an $\ell$-replacement for Phragmms can be computed in time $\mathcal{O}(1)$, assuming that we get as input the value $\min\{s_1,\dots,s_{k-\ell}\}$, where $s_i$ is the Phragmms score of candidate $c_i$ for any $i \in [k]$.  
\end{theorem}

\begin{proof}
    For the definition of candidate scores in Phragmms, we refer to~\citet{CeSt21a}. 
    For election $(C,V,A,w,k)$, recall that any $\ell$-replacement $(C',V',A',w')$ is ``properly separated,'' in the sense that $C\cap C'=\emptyset$ and $A\cap A'=\emptyset$. 
    As a consequence (and similarly to the previous election rules), the outcome of Phragmms can be computed by merging its executions on each instance separately. 
    More precisely, if $c_1, \dots, c_{k}$ are the elected candidates in $(C,V,A,w,k)$ labeled by election order, and $s_i$ is the score of $c_i$ at the time of its election, 
    and similarly if $(c'_1,s'_1), \dots, (c'_\ell,s'_\ell)$ are the candidates and their scores when we run the election $(V',C',A', w', \ell)$, 
    then we claim that in an execution of Phragmms on the extended instance $(V \cup V',C \cup C',A \cup A', w \cup w',k)$, 
    the selection order is determined by starting with the empty committee and repeating the following process $k$ times:
    Let $i \in [k]$ be the smallest index of an unelected candidate in $\{c_1,\dots,c_{k}\}$ 
    and let $j \in [\ell]$ be the smallest index of an unelected candidate in $\{c'_1, \dots, c'_{\ell}\}$; 
    if $s_i > s'_j$, elect $c_i$, otherwise elect $c_j$.
    To prove the claim, it suffices to verify that for each subinstance of the extended instance, neither the updating of the scores nor the rebalancing step of the vote assignment 
    is influenced by the other subinstance. This fact can be easily checked when going through Section 4 in~\citet{CeSt21a}. 

    Having established this claim, we aim to find a weight-minimal $\ell$-replacement. We first show that we can assume without loss of generality, that such a replacement has a simple structure, namely, consists of one voter who approves all of the $\ell$ candidates in $C'$. To see why, consider any $\ell$-replacement of some weight $x:= \sum_{v \in V'} w'(v)$, and let $s'_1, \dots, s'_\ell$ be the corresponding score sequence. Lemma 18 by \citet{CeSt21a} states that, at every iteration $i$, the minimum backing weight of the vote assignment after iteration $i-1$ is at least as the score $s'_i$. Moreover, since no backing weight is reduced below $s'_i$ in iteration $i$, this implies an upper bound of $x/i$ on $s'_i$. 
    Considering the trivial $\ell$-replacement of weight $x$, i.e., one voter of weight $x$ approving all candidates in $C'$, we observe that the resulting score sequence corresponds exactly to $x/1,x/2,\dots,x/\ell$.
    Therefore, the existence of some $\ell$-replacement of weight $x$ implies the fact that there also exists such a trivial $\ell$-replacement of weight $x$. 

    It remains to determine the minimum weight of some $\ell$-replacement. For now, we call this weight $x$. Any $\ell$-replacement is successful if and only if candidate $c'_{\ell}$ gets elected before candidate $c_{k-\ell+1}$. This is the case if and only if $s'_{\ell} = \frac{x}{\ell} \geq \min\{s_1,\dots,s_{k-\ell+1}\}$. 
    Thus, a minimum $\ell$-replacement has weight $$z = \ell \cdot \min\{s_1,\dots,s_{k-\ell+1}\},$$ proving the claim. Lastly, we remark that due to the last argument, the following $\ell$-replacement is successful as well: There are $\ell$ voters, each with weight $\frac{z}{\ell}$ and each approving a distinct candidate in $C'$.
\end{proof}

\paragraph{Sequential PAV} Maybe surprisingly, the task of computing the minimum exogenous cost for replacing $\ell$ candidates for seq-PAV leads to an interesting theoretical question, which may be of independent interest for future work. In the following, we report on our observations so far, and provide a linear program with which we were able to solve the question for all $\ell \le 21$. 

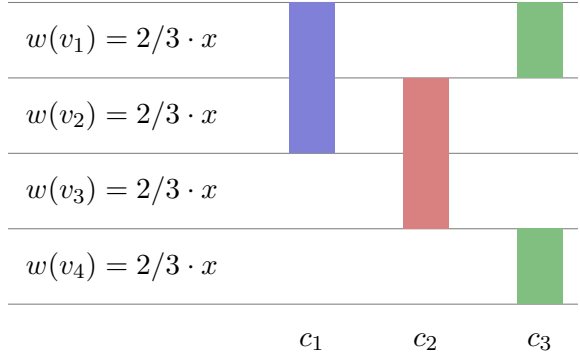
\begin{figure}[t!]
    \centering
    \begin{tikzpicture}
        \node at (-1,0) {$w(v_4) = 2/3 \cdot x$}; 
        \node at (-1,1) {$w(v_3) = 2/3 \cdot x$}; 
        \node at (-1,2) {$w(v_2) = 2/3 \cdot x$}; 
        \node at (-1,3) {$w(v_1) = 2/3 \cdot x$}; 
        \draw[black!50] (-2.5,0.5) -- (5,0.5); 
        \draw[black!50] (-2.5,1.5) -- (5,1.5); 
        \draw[black!50] (-2.5,2.5) -- (5,2.5); 
        \draw[black!50] (-2.5,3.5) -- (5,3.5); 
        \draw[black!50] (-2.5,-.5) -- (5,-.5); 
        \node at (1.5,-1) {$c_1$}; 
        \node at (3,-1) {$c_2$}; 
        \node at (4.5,-1) {$c_3$}; 
        \filldraw [fill=blue!70!black!50!white,draw=none] (1.2,3.5) rectangle (1.8,1.5);
        \filldraw [fill=red!70!black!50!white,draw=none] (2.7,0.5) rectangle (3.3,2.5);
        \filldraw [fill=green!50!black!50!white,draw=none] (4.2,-.5) rectangle (4.8,0.5);
        \filldraw [fill=green!50!black!50!white,draw=none] (4.2,3.5) rectangle (4.8,2.5);
    \end{tikzpicture}
    \caption{Illustration of an optimal $3$-replacement for Seq-PAV, with a total weight $\frac{8}{3} x$, in which the $3^{\text{rd}}$ candidate has marginal sequential PAV score of $x$. The approval ballots are interpreted as follows: $A_{v_1} = \{c_1,c_3\}, A_{v_2} = \{c_1,c_2\}, A_{v_3} = \{c_2\}$ and $A_{v_4} = \{c_3\}$.}
    \label{fig:SeqPAV-example}
\end{figure}

For a given election $(C,V,A,w,k)$, let $c_i$ be the $i^{\text{th}}$ candidate that is selected by seqPAV and let $s_i$ be the marginal PAV score of candidate $c_i$ at the time of selection. More precisely, this means that $$s_i = \text{sc}(\{c_1,\dots,c_{i}\}) - \text{sc}(\{c_1,\dots,c_{i-1}\}),$$ where we recall that the function $\text{sc}$ is defined as $$\text{sc}(W) = \sum_{v \in V} w(v) \sum_{j=1}^{|A_v \cap W|} \frac{1}{j}.$$ Similarly, for any (potential) $\ell$-replacement $(C',V',A',w')$, let $c'_i$ be the candidate selected at point $i \in [\ell]$ in the election $(C',V',A',w')$ and $s'_i$ be the marginal PAV score of this candidate at the point of selection. Then, the $\ell$-replacement is successful, if and only if $s_{k-\ell +1} < s'_{\ell}$. The reason for this is twofold: (1) Both sequences $s_1, \dots, s_k$ and $s'_1, \dots, s'_{\ell}$ are monotone. (2) Clearly, when running the extended election, the two subelections behave completely independently from each other. 

Hence, the question of computing the minimum exogenous cost for replacing $\ell$ candidates reduces to the following task: Given $\ell \in \mathbb{N}$ and $x:=s_{k-\ell+1} \in \mathbb{R}$, find an election $(C',V',A',w',\ell)$ of minimum weight, such that the marginal contribution of the $\ell^{\text{th}}$ selected candidate is at least $x$. For all previous voting rules, we were always able to restrict ourselves to $\ell$-replacments with a very simple structure, e.g., one voter approving all candidates in $C'$. Let us try the same by considering two natural simple election structures in the context of seq-PAV. In the first, consider one voter $v$ with weight $w(v)$ approving all $\ell$ candidates in $C'$. Then, the score sequence corresponds to $s'_1 = w(v), \dots, s_{\ell}=\frac{w(v)}{\ell}$. Hence the minimum successful weight when we restrict ourselves to this election format is $w(v) = \ell \cdot x$. Similarly, if we consider $\ell$ voters, all of which approve one of the candidates in $C'$ each, then the score sequence corresponds to $s'_1=w(v_1), \dots,s'_{\ell}=w(v_{\ell})$, hence, the minimum weight also corresponds to $\sum_{v \in V'} w(v) = \ell \cdot x$. Given those two extremes, one might be tempted to conjecture that $\ell \cdot x$ is the correct answer, however, this is not the case since there can be non-trivial synergy effects between the candidates. For intuition, we give an example for $\ell = 3$ in \Cref{fig:SeqPAV-example}. 

When constructing minimum $\ell$-replacements for seq-PAV, there is a non-trivial trade-off between reusing the (reduced) approval weight of already selected candidates and not violating the greedy selection order. In the following, we describe a linear program\footnote{The underlying ideas of this LP are similar to those used by~\citet{SFF+17a} to prove that seq-PAV satisfies JR when $k \leq 5$, but fails it when $k \geq$ 6.} that optimally solves this trade-off, albeit having exponentially many variables. Before we define the LP, we define for any $S \subseteq C'$ and $i \in [\ell]$ the value $$t(S,i) = (|S \cap \{c_1, \dots, c_{i-1}\}| + 1)^{-1}.$$ The value corresponds to the marginal contribution of a voter to the PAV score at the moment when candidate $c_i$ is elected, assuming that the voter has weight one and approves all candidates in $S$. The LP has one variable per subset $S \subseteq C'$, called $y(S)$, which we interpret as the weight of the voter approving exactly the candidates in $S$. Note that this entirely describes any election with candidate set $C'$. We arbitrarily index the elements in $C'$ by $c_1, \dots, c_{\ell}$ and the first set of inequalities of the LP enforces that the candidates are elected by seq-PAV in this order. That is, for any pair $i,j \in [\ell]$ with $i < j$, the marginal contribution of $i$ (at iteration $i$) is at least as large as the marginal contribution of $j$ at iteration $i$. Lastly, the other non-trivial constraint implies that candidate $c_{\ell}$ has marginal contribution at least $x$ in iteration $\ell$. 

\begin{align*}
    \text{min} & \sum_{S \subseteq C'} y(S) \\
     \sum_{S \subseteq C', c_i \in S} t(S,i) \cdot y(S) & \geq \sum_{S \subseteq C', c_j \in S} t(S,i) \cdot y(S) \\ & \quad \quad \quad \quad \quad \;\; \forall \; i,j \in [\ell], j > i\\ 
     \sum_{S \subseteq C', c_{\ell} \in S} t(S,\ell) \cdot y(S) & \geq x \\
     y(S) & \geq 0 \quad \quad \quad \quad \forall \; S \subseteq C'
\end{align*}

\begin{figure}[t!]
    \centering
    \includegraphics[scale=0.5]{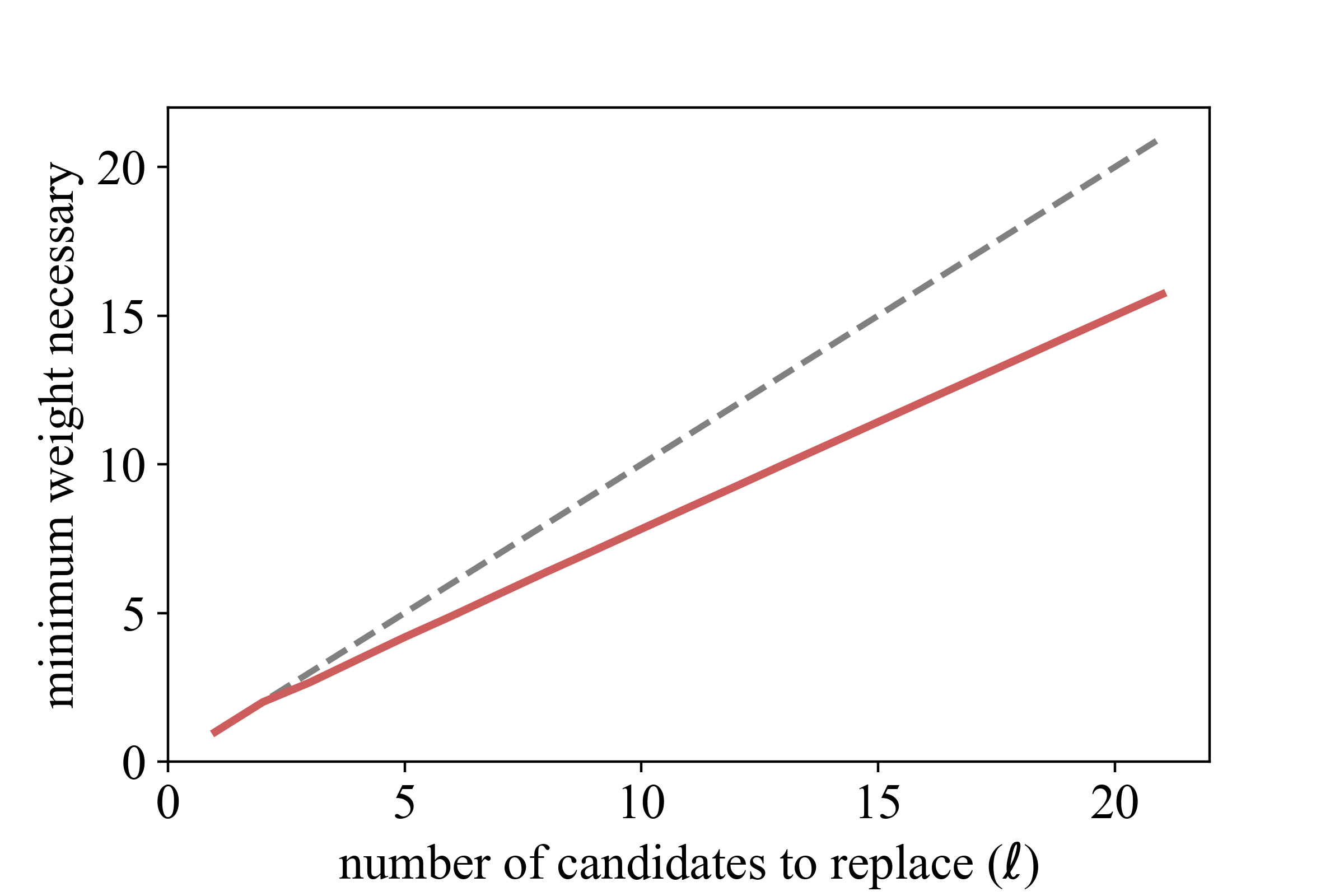}
    \caption{Minimum exogenous cost for replacing $\ell$ candidates when normalizing the marginal PAV score of the $k -\ell +1^{\text{th}}$ candidate in the original election to be one. The grey dashed line indicates the identity, which corresponds to the minimum exogenous cost if we restrict ourselves to $\ell$-attacks with simple structures (e.g., those with only one voter or those with voters approving one candidate each).}
    \label{fig:seqPAV-LP}
\end{figure}

Note that this LP is invariant to scaling $x$ since an optimal solution $y \in \mathbb{R}^{[0,1]^\ell}$ for input $x =1$ lets us derive an optimal solution $y'$ for any $x'$ by setting $y'(S) = x' \cdot y(S)$ for all $ S\subseteq C'$. Hence, we restricted ourselves to the case $x=1$ and computed the solution of the LP for all values $\ell \in [21]$. We report these values in \Cref{fig:seqPAV-LP}. 

The formulation as an LP also provides us with a non-trivial structural insight. That is, there always exists an optimal $\ell$-replacement, with at most $\frac{\ell(\ell-1)}{2} + 1$ voters. We obtain this bound by the following observation: The LP has $2^{\ell}$ variables, $2^{\ell}$ ``trivial'' constraints (those of form $y(S)\geq 0$) and $\frac{\ell (\ell -1)}{2} + 1$ non-trivial constraints. Now, for any extreme point of the LP, there have to be at least $2^{\ell}$ tight constraints. Hence, at least $2^{\ell} - \frac{\ell (\ell -1)}{2} + 1$ of the trivial constraints have to be tight, implying an upper bound of at most $\frac{\ell(\ell-1)}{2} + 1$ non-zero variables. Nevertheless, finding these variables (corresponding to voters) may still be a challenging task. 

Lastly, we remark that the LP might be solvable in polynomial time by a standard \emph{column generation} approach. That is, since our LP has polynomial many constraints, deriving the dual LP results in an LP with polynomial many variables and exponential many constraints. Now, if we could find an algorithm for the corresponding \emph{separation problem} (i.e., given a solution to the dual LP, determine in polynomial time whether a constraint is violated), then we could employ the Ellipsoid method, to solve the problem in polynomial time. In fact, we derived the dual LP and considered the separation problem, however, we were not able to find a polynomial-time algorithm nor show hardness of the separation problem. 
 
\paragraph{Method of Equal Shares} For MES, the situation is less clear, as the initial budget of a voter depends on the total weight of voters. Thus, by adding new voters the whole execution of the rule might change. Even worse, MES is not candidate monotone with additional voters \citep{LaSk22a}, which means that adding an additional supporter for a currently selected candidate can actually lead to the candidate no longer being selected.

\begin{table*}[t!]
    \centering
    \resizebox{0.8\textwidth}{!}{
    
    \begin{tabular}{l|cc|cc|cc|cc|cc|ccccccc}
\toprule
{} &  \multicolumn{2}{c|}{PAV} &  \multicolumn{2}{c|}{max. support} &  \multicolumn{4}{c|}{min. average satisfaction} & \multicolumn{2}{c|}{priceability}   \\
{} &   \multicolumn{2}{c|}{score} & \multicolumn{2}{c|}{loser} &  \multicolumn{2}{c|}{$\ell=1$} & \multicolumn{2}{c|}{$\ell=5$}  &  \multicolumn{2}{c}{gap}  & \\
\midrule\midrule
      & {\footnotesize $k=300$} & {\footnotesize P-APP} &  {\footnotesize $k=300$} & {\footnotesize P-APP} &{\footnotesize $k=300$} & {\footnotesize P-APP}& {\footnotesize $k=300$} & {\footnotesize P-APP}& {\footnotesize $k=300$} & {\footnotesize P-APP} \\ 
      & {\footnotesize $k=200$} &  {\footnotesize $k=400$} & {\footnotesize $k=200$} &  {\footnotesize $k=400$} & {\footnotesize $k=200$} &  {\footnotesize $k=400$}& {\footnotesize $k=200$} &  {\footnotesize $k=400$}& {\footnotesize $k=200$} &  {\footnotesize $k=400$} \\
      & {\footnotesize $k=250$} &  {\footnotesize $k=350$} & {\footnotesize $k=250$} &  {\footnotesize $k=350$} & {\footnotesize $k=250$} &  {\footnotesize $k=350$}& {\footnotesize $k=250$} &  {\footnotesize $k=350$}& {\footnotesize $k=250$} &  {\footnotesize $k=350$} \\\midrule\midrule
\multirow{3}{*}{AV}           & 2.502  &   & 0.011 &   & 0.000 & & - & & 6.21 &   \\
       & 2.075  &  2.690 & 0.021 & 0.0038  & 0.00 & 0.120& 3.357 & - &9.421  &  3.503 \\
       & 2.324 & 2.624&0.017 &0.0066  & 0 & 0 & $-$ & $-$ & 7.621 & 6.301\\ \midrule
\multirow{3}{*}{SAV}            & 2.547 &   & 0.022  &   & 0.011 & & 4.359 & & 7.14 &   \\
       & 2.082  & 2.715  & 0.035  &  0.006 & 0.00 &3.145 & 0.79 & - &10.234  &  1.854 \\ 
       & 2.348 & 2.658&0.027 & 0.0120 & 0.0014 & 0.67 &0.920 &12.02 & 8.886 &4.464 \\\midrule
      \multirow{3}{*}{seq-PAV}           & 2.585  & 2.803   &0.028  & 0.084  &1.092 &0.86 & 7.442 & 5.171 &0.53  & -0.032  \\
       & 2.284 &  2.718 & 0.044 &  0.0076 & 0.99 &5.069 & 6.677 &14.277 & 0.640 &  -0.240 \\ 
       &2.452 & 2.673& 0.042&0.0150 &1.05 &2.123 &6.722 &11.08 &0.63 &0.15 \\\midrule
      \multirow{3}{*}{Phragmms}            & 2.578  & 2.777  & 0.039 &  0.084 & 0.960 &0.56 & 6.840 & 4.610& -0.12  & -0.033  \\
       & 2.271 & 2.712  & 0.043 &  0.0110 & 0.74 & 1.901& 6.652 & 9.826& -0.083 &  -0.550 \\ 
       &2.441 & 2.666& 0.043 &0.0200  & 0.84 &1.047 & 6.315&10.032  &-0.042 &-0.210 \\\midrule
      \multirow{3}{*}{seq-Phrag.}            &    2.578 & 2.802&0.039  & 0.084  & 0.990 & 0.75 & 6.678 & 5.180&  -0.16&  -0.035 \\
       & 2.277 & 2.713  & 0.043 &  0.0110 & 0.91 &2.422 & 6.324 &11.305 & -0.078 & -0.580  \\ 
       &2.443 & 2.667&0.043 &0.0210 &1.0010 &1.111 &6.25 & 9.635&-0.044 &-0.240\\\midrule
      \multirow{3}{*}{MES}           & 2.579    &2.777& 0.039  & 0.084 & 0.980 & 0.71 & 6.773 & 5.110 & -0.17 & -0.033\\
       & 2.276  &  2.713 & 0.043 &  0.0110 & 0.94 &2.460 & 6.316 & 11.238&-0.077  & -0.590  \\
       & 2.443& 2.668&0.043 &0.0210 &0.99 & 1.12&6.311  &9.731 &-0.042 & -0.24\\ \midrule      
\bottomrule
\end{tabular}    
}
    \caption{Average values for different notions related to preventing underrepresentation. }
    \label{tab:underreptrends}
\end{table*}

\begin{table*}[t!]
    \centering
    
    \resizebox{1\textwidth}{!}{
    \begin{tabular}{lcc|cc|cc|cc|cc|cc|cc}
\toprule
{} &  \multicolumn{2}{c|}{maximin} &  \multicolumn{2}{c|}{maximin} &  \multicolumn{2}{c|}{min. support} & \multicolumn{4}{c|}{cost of replacing} & \multicolumn{4}{c|}{min. approval weight}   \\
{} &   \multicolumn{2}{c|}{value} & \multicolumn{2}{c|}{variance} &  \multicolumn{2}{c|}{winner} & \multicolumn{2}{c|}{one-third}  &  \multicolumn{2}{c|}{half}  & \multicolumn{2}{c|}{one-third}  &  \multicolumn{2}{c|}{half} \\ \midrule\midrule
      & {\footnotesize $k=300$} & {\footnotesize P-APP} &  {\footnotesize $k=300$} & {\footnotesize P-APP} &{\footnotesize $k=300$} & {\footnotesize P-APP}& {\footnotesize $k=300$} & {\footnotesize P-APP}& {\footnotesize $k=300$} & {\footnotesize P-APP}& {\footnotesize $k=300$} & {\footnotesize P-APP}& {\footnotesize $k=300$} & {\footnotesize P-APP} \\ 
      & {\footnotesize $k=200$} &  {\footnotesize $k=400$} & {\footnotesize $k=200$} &  {\footnotesize $k=400$} & {\footnotesize $k=200$} &  {\footnotesize $k=400$}& {\footnotesize $k=200$} &  {\footnotesize $k=400$}& {\footnotesize $k=200$} &  {\footnotesize $k=400$}& {\footnotesize $k=200$} &  {\footnotesize $k=400$}& {\footnotesize $k=200$} &  {\footnotesize $k=400$} \\
      & {\footnotesize $k=250$} &  {\footnotesize $k=350$} & {\footnotesize $k=250$} &  {\footnotesize $k=350$} & {\footnotesize $k=250$} &  {\footnotesize $k=350$}& {\footnotesize $k=250$} &  {\footnotesize $k=350$}& {\footnotesize $k=250$} &  {\footnotesize $k=350$}& {\footnotesize $k=250$} &  {\footnotesize $k=350$}& {\footnotesize $k=250$} &  {\footnotesize $k=350$} \\\midrule\midrule
\multirow{3}{*}{AV}           & 0.0015 &   & 0.0047&   & 0.0110 &  &   0.14  &  & 0.27 & & 0.22 & & 0.36 & \\
      & 0.0018 & 0.00076  & 0.0060 &  0.0056 & 0.0210 & 0.0038 &0.14  & 0.14 & 0.24  & 0.27 & 0.16 & 0.22 & 0.26 & 0.35 \\ 
      & 0.0017 &0.00087 &0.0049 & 0.0057& 0.017 &0.0066 & 0.15 &0.14 & 0.24 & 0.26& $-$ & $-$ &$-$ &$-$ \\\midrule
\multirow{3}{*}{SAV}           &0.0018  &   & 0.0043 &   &0.0024  &  & 0.24 &  & 0.4 &  & 0.24 & & 0.38 &  \\
       &0.0026  & 0.00068   & 0.0051& 0.0060  & 0.0028 & 0.0017 & 0.18 &  0.25 & 0.29 & 0.48 & 0.18 & 0.23 &  0.29 &  0.38  \\ 
      &0.0023 &0.00120 &0.0043 & 0.0049&0.0028 &0.002 &0.21  &0.26 &0.35  &0.45 & $-$& $-$&$-$ &$-$ \\ \midrule
\multirow{3}{*}{seq-PAV}           & 0.0024 & 0.003  & 0.0031& 0.00063  & 0.0028 & 0.0032 & ? & ? & ? & -?& 0.26 & 0.31 & 0.42 & 0.47  \\
       &  0.0038  & 0.00074  & 0.0031&  0.0066 & 0.0051 &0.0018  & ? &  ? & ? & ? & 0.28 & 0.23 & 0.44 & 0.38 \\
       &0.0030 &0.0014 &0.0030 &0.0044 & 0.0032 & 0.002 & ? & ? & ? & ? &$-$ &$-$ &$-$ &$-$ \\ \midrule
\multirow{3}{*}{Phragmms}           &0.00272  & 0.003  & 0.0023& 0.00056  & 0.0028 & 0.0032 & 0.4 & 0.42  & 0.79 &  0.82 & 0.28 & 0.31 & 0.44 & 0.46 \\
       &0.0040 &  0.00140 & 0.0023& 0.0040  & 0.0058 & 0.0020 & 0.38 &  0.4 & 0.74 & 0.8 & 0.29 & 0.25 & 0.45 & 0.38   \\
       & 0.0033 &0.002 & 0.0021 &0.0032 & 0.0036 &0.0027  & 0.39 &0.4 & 0.77 &0.79 & $-$&$-$ &$-$ &$-$ \\ \midrule
\multirow{3}{*}{seq-Phag.}           & 0.0027 & 0.003  & 0.0024& 0.00066   & 0.0028 & 0.0032 & 0.4 & 0.42  & 0.79 &  0.82 & 0.28 & 0.31 & 0.44 &  0.47 \\
       & 0.0041 & 0.00120  &0.0021 & 0.0042  & 0.0058 &  0.0020 & 0.39 &  0.4 & 0.75 & 0.8 & 0.29 & 0.25 & 0.45 & 0.38  \\ 
       & 0.0033 & 0.0019 &0.0021 &0.0033 &0.0036 &0.0026 &0.4 &0.4  &0.78 &0.8 & & & & \\ \midrule
\multirow{3}{*}{MES}           & 0.00269 & 0.003  & 0.0024& 0.00055  & 0.0028 & 0.0032 & ? & ? & ? & ? & 0.28 & 0.31 & 0.44 & 0.47  \\
       & 0.0041 & 0.00120  & 0.0021& 0.0042  & 0.0057 & 0.0019 &  ? & ? &  ? & ? & 0.28 & 0.25 & 0.45 & 0.38 \\ 
       &0.0032  &0.0019 &0.0022 &0.0033 &0.0036 &0.0026 & ? & ? & ? & ? & $-$&$-$ &$-$ &$-$ \\ \midrule
      
\bottomrule
\end{tabular}   } 

    \caption{Average values for different notions related to preventing overrepresentation. Entries marked with  ``$-$'' were not computed due to a shortage of computation time. ``Maximin variance'' stands for the variance of the candidate's backing weight in the maximin assignment. }
    \label{tab:overreptrends}
\end{table*}

\section{Additional Material for \Cref{sec:desdec}}
\label{app:desdec}

\Cref{tab:underreptrends} (on page \pageref{tab:underreptrends}) shows some of our measures regarding the prevention of underrepresentation. We compare different committee sizes $k\in \{200,250,300,350,400\}$ and whether candidates can be selected multiple times for $k=300$ (following the terminology of \citet{BGP+23a}, we call this setting ``P-APP'' setting; also see their work for a formal definition). 
\Cref{tab:overreptrends} shows the same comparison for measures regarding overrepresentation.  

\paragraph{Changing the Committee Size.}
All rules are similarly affected by changes in the committee size, which is why we only report  general trends focusing on the reported average values. 
Increasing the committee size from $300$ to $400$ typically leads to 
\begin{itemize}
    \item a small ($\approx7.5\%$) increase of the PAV score,
    \item a substantial (typically at least $\approx100\%$) increase of the minimum average satisfaction, 
    \item a substantial decrease in the priceability gap, 
    \item a decrease of the maximum support value by typically at least around $50\%$, 
    \item no significant change in terms of the cost of replacing, and 
    \item a small ($\approx5\%$) decrease of the minimum approval weight of size-$\frac{k}{3}$ and size-$\frac{k}{2}$ subsets of the committee. 
\end{itemize}

Decreasing the committee size from $300$ to $200$ typically leads to 
\begin{itemize}
    \item a at least $11\%$ decrease of the PAV score,
    \item small drops in the average satisfaction for the proportional rules, 
    \item a small increase in the priceability gap, 
    \item an around $\approx50\%$ increase of the maximin value for the proportional rules, yet  no strong increase for AV,
    \item no significant change of the exogenous cost of replacing (except for SAV where values drop by $\approx25\%$), and 
    \item no significant change of the minimum approval weight of size-$\frac{k}{3}$ and size-$\frac{k}{2}$ subsets of the committee, except for AV and SAV where values drop by around $25\%$.     
\end{itemize}

Results for $k=350$ almost always lie in between those for $k=300$ and $k=400$, and results for $k=350$ in between those for $k=200$ and $k=300$. In particular, all reported measures are either not strongly affected by the committee size or behave monotonically. 

\paragraph{Selecting Candidates Multiple Times.}
AV and SAV do not admit natural translations to the P-APP setting, which is why we only report results on the four proportional rules here.
Notably, for the definition of $\ell$-supporting groups and priceability, all candidates are viewed as being non-selected (as we could always include an additional copy of a candidate in the committee). 
In general, we observe that allowing candidates to be selected multiple times leads to 
\begin{itemize}
    \item an increase in the PAV score (even when compared to the values for $k=400$),
    \item a decrease in the minimum average satisfaction of $\ell$-supporting groups (which, as discussed in the main body, can be explained because now also candidates which are included in the committee count as non-selected), 
    \item an increase in the priceability gap (again same explanation as for $\ell$-supporting groups), 
    \item an increase of the maximin support value (which is roughly at the level as for $k=250$), 
    \item a substantial decrease of the maximin variance by around $75\%$, which means that the backing weight of candidates is much more uniform here, and 
    \item a small increase in the cost of replacing and in the minimum approval weight of size-$\frac{k}{3}$ and size-$\frac{k}{2}$ subsets of the committee.
\end{itemize}
\FloatBarrier

\section{Kusama Elections} \label{kusama}

\begin{table*}[!t]
    \centering
    \begin{minipage}{0.34\textwidth}
  \centering
   \includegraphics[width=0.9\textwidth]{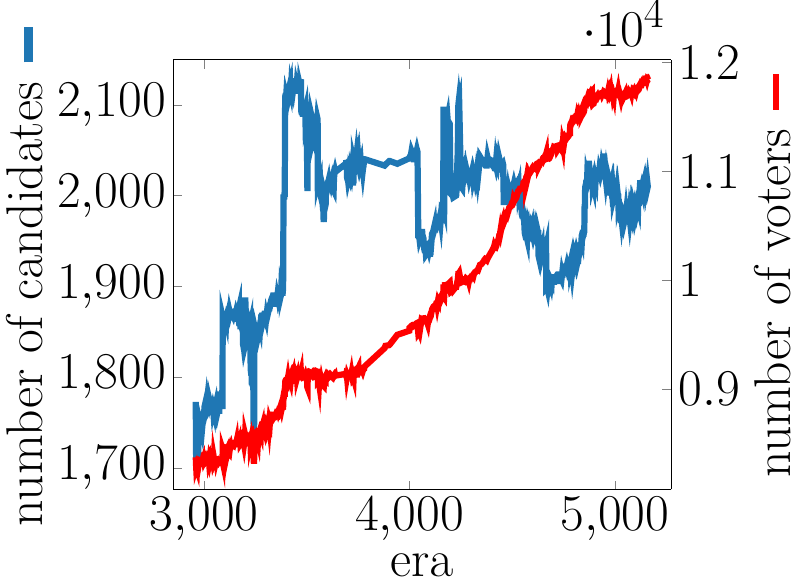}
  \captionof{figure}{Number of candidates (blue) and voters (red) in our $1520$ Kusama elections.}
\label{fig:basic1-app}
    \end{minipage} \hfill
    \begin{minipage}{.64\textwidth}
    \centering
  \resizebox{0.95\textwidth}{!}{\begin{tabular}{lcccccc}
\toprule
{} &       AV &      SAV &   seq-PAV &  Phragmms &  seq-Phrag. &      MES \\
\midrule
AV          &    $-$ &  913.50 &  912.65 &    931.05 &       916.75 &  913.65 \\
SAV         &  913.50 &   $-$ &  972.05 &    938.05 &       945.35 &  944.60 \\
seq-PAV      &  912.65 &  972.05 &    $-$ &    941.40 &       956.85 &  957.30 \\
Phragmms    &  931.05 &  938.05 &  941.40 &      $-$ &       951.85 &  948.10 \\
seq-Phrag. &  916.75 &  945.35 &  956.85 &    951.85 &         $-$ &  984.60 \\
MES        &  913.65 &  944.60 &  957.30 &    948.10 &       984.60 &    $-$\\
\bottomrule
\end{tabular}}
    \caption{Average overlap between committees returned by different voting rules. The committee size is $1000$.}
    \label{tab:overlapKus}
    \end{minipage}
\end{table*}

\begin{table*}[!t]
    \centering
  \resizebox{\textwidth}{!}{   \begin{tabular}{lcccccccccc}
\toprule
{} &  weighted & PAV  &  max. approval weight  & JR  &  EJR+  &   \multicolumn{3}{c}{min. average satisfaction} &  priceability  & priceability  \\
{} &  satisfaction &   score &  loser &   violations &   violations &   $\ell=1$ &  $\ell=5$ &  $\ell=10$ &   violations &  gap \\
\midrule
AV          &                 14.966 &      3.042 &                    0.0042 &          35.95 &             76.4 &            0.000 &            1.479 &               $-$ &                   148.25 &             7.225 \\
SAV         &                 14.608 &      3.084 &                    0.0100 &          20.25 &             43.5 &            0.430 &            3.995 &            15.761 &                   137.90 &             9.702 \\
seq-PAV      &                 14.553 &      3.102 &                    0.0120 &           0.00 &              0.0 &            2.178 &            9.192 &            17.041 &                    64.45 &             4.358 \\
Phragmms    &                 14.586 &      3.090 &                    0.0150 &           0.00 &              0.0 &            1.113 &            7.092 &            14.250 &                     0.00 &            -0.060 \\
seq-Phrag. &                 14.508 &      3.094 &                    0.0150 &           0.00 &              0.0 &            1.526 &            7.620 &            14.514 &                     0.00 &            -0.095 \\
MES        &                 14.476 &      3.094 &                    0.0150 &           0.00 &              0.0 &            1.526 &            7.738 &            14.728 &                     0.00 &            -0.071 \\
\bottomrule
\end{tabular}  }
    \caption{Average values for different notions related to preventing underrepresentation.}\label{underrep-kusama}
\end{table*}

\begin{table}[!t]
    \centering
    \begin{tabular}{llllllr}
\toprule
{} 			   &   \multirowcell{2}[0pt][l]{maximin\\support\\value} &	\multirowcell{2}[0pt][l]{min. appr.\\weight of\\winner} 	& \multicolumn{3}{c}{cost of ``replacing'' $\ell$} \\
\cmidrule{4-6   }
{} 			      &   		   		&  		&  	$\ell=1$ & $\ell=\frac{k}3$ & $\ell=\frac{k}2$ \\
\addlinespace[0.5em] \midrule
AV          &           0.00047 &         0.00420                &        0.0042 &0.22 &       0.38   \\
SAV         &           0.00055 &        0.00079                &       0.00047 &  0.24&       0.41              \\
seq-PAV      &           0.00062 &         0.00079                &     ? & ? &     ?  \\
Phragmms    &           0.00080 &         0.00082                &       0.0008 & 0.43 &      0.85   \\
seq-Phrag. &           0.00078 &         0.00081                &       0.00078 &  0.43  &    0.86  \\
MES         &           0.00079 &         0.00082                &           ? &           ? &                ?  \\
\bottomrule
\end{tabular}
    \caption{Measures related to overrepresentation.} 
\label{tab:overrep-kusama}
\end{table}

\begin{figure*}[t]
  \centering
  \resizebox{0.33\textwidth}{!}{
\begin{tikzpicture}[every plot/.append style={line width=2.2pt}]

\definecolor{crimson2143940}{RGB}{214,39,40}
\definecolor{darkgray176}{RGB}{176,176,176}
\definecolor{darkorange25512714}{RGB}{255,127,14}
\definecolor{forestgreen4416044}{RGB}{44,160,44}
\definecolor{mediumpurple148103189}{RGB}{148,103,189}
\definecolor{orchid227119194}{RGB}{227,119,194}
\definecolor{sienna1408675}{RGB}{140,86,75}
\definecolor{steelblue31119180}{RGB}{31,119,180}

\begin{axis}[
legend columns=2, 
legend cell align={left},
legend style={
  fill opacity=0.8,
  draw opacity=1,
  draw=none,
  text opacity=1,
  at={(0.5,1.35)},
  line width=3pt,
  anchor=north,
   /tikz/column 2/.style={
  	column sep=10pt,
  }, font=\LARGE
},
legend entries={AV,
	seq-Phrag.,
	SAV, Phragmms,
	seq-PAV, MES},
tick align=outside,
tick pos=left,
x grid style={darkgray176},
xlabel={$\ell$-supporting groups},
xmin=0.4, xmax=20.6,
xtick style={color=black},
y grid style={darkgray176},
ylabel={minimum average satisfaction},
ymin=-0.749999973283545, ymax=25,
ytick style={color=black},every tick label/.append style={font=\LARGE}, 
label style={font=\LARGE},
xticklabels={$2$,$4$,$6$,$8$,$10$,$12$,$14$,$16$,$18$,$20$},
xtick={2,4,6,8,10,12,14,16,18,20}
]
\addlegendimage{steelblue31119180}
\addlegendimage{mediumpurple148103189}
\addlegendimage{darkorange25512714}
\addlegendimage{crimson2143940}
\addlegendimage{forestgreen4416044}
\addlegendimage{sienna1408675}
\addplot [semithick, steelblue31119180, mark=x, mark size=3, mark options={solid}]
table {%
1 0
2 0.0320582094820862
3 0.100884185627793
4 1.3378748932156
5 1.47854490530676
6 4.00457676600128
};
\addplot [semithick, darkorange25512714, mark=x, mark size=3, mark options={solid}]
table {%
1 0.426712961381568
2 0.738366725643202
3 1.04924167989272
4 2.78572421310874
5 3.9948961659092
6 4.6262125450899
7 5.95852418842838
8 7.45463174063999
9 11.2120299940962
10 15.7612535330499
11 17.512948569497
12 15.376869768742
};
\addplot [semithick, forestgreen4416044, mark=x, mark size=3, mark options={solid}]
table {%
1 2.17780151912326
2 3.90087921819969
3 5.05736375686208
4 8.0350123370656
5 9.19153557790087
6 11.0026125551056
7 12.7767729503674
8 14.3180918188495
9 15.5911190152058
10 17.0408316254272
11 17.570172101033
12 19.96143687761
13 20.7205123609467
14 21.5744371336148
15 21.7406110618454
};
\addplot [semithick, crimson2143940, mark=x, mark size=3, mark options={solid}]
table {%
1 1.11294958219491
2 2.36119910503014
3 3.41184106576077
4 5.2793063095964
5 7.09159053466472
6 8.49917607069921
7 9.29862495169537
8 10.5493272951082
9 13.0996885914568
10 14.2497381799667
11 14.8114911930668
12 15.8647294223104
13 16.3322113135881
14 17.7130870722144
15 20.1229197655287
16 20.9999872864703
17 21.2856992954536
18 21.8571236832022
19 22.2499825408481
20 22.4999917444124
};
\addplot [semithick, mediumpurple148103189, mark=x, mark size=3, mark options={solid}]
table {%
1 1.52571327923654
2 2.89845118168837
3 3.72074458630606
4 6.47945912906711
5 7.62022468323707
6 8.64263872833054
7 9.43380379366836
8 10.6481662655278
9 12.7196117016598
10 14.5137721642375
11 14.9993450347664
12 16.0778039760858
13 16.6287956903515
14 17.9219101012017
15 20.3214977909882
16 21.4999204936848
17 21.8571276323752
18 22.4285518320538
19 22.3333103218485
20 22.9999843413198
};
\addplot [semithick, sienna1408675, mark=x, mark size=3, mark options={solid}]
table {%
1 1.52570123712232
2 3.04843345703685
3 3.91655089368499
4 6.53356801893568
5 7.73843186986284
6 8.7427942828263
7 9.51171679852604
8 10.5984717152353
9 12.8557572542096
10 14.7282532205942
11 15.1082926854797
12 16.0240334352458
13 16.5333145307315
14 17.7623135119977
15 20.1022872909239
16 21.4984238212304
17 21.8554416816122
18 22.2856782784463
19 22.3332717811558
20 22.9999843413198
};
\addplot [line width=1.5pt, black,dashed]
table {%
1 0
2 1
3 2
4 3
5 4
6 5
7 6
8 7
9 8
10 9
11 10
12 11
13 12
14 13
15 14
16 15
17 16
18 17
19 18
20 19
};
\end{axis}

\end{tikzpicture}}
  \caption{Minimum average satisfaction of $\ell$-supporting
groups. The dashed line is the function $f(\ell)=\ell-1$. Lines
stop in case no $\ell$-supporting group of this size exists.}\label{average-sati-kusama}
\end{figure*}

\begin{figure*}[tb]
\centering
\begin{minipage}[t]{.48\textwidth}
  \centering
  \resizebox{0.75\textwidth}{!}{\input{./plots_kusama/lost_stake_maximin}}
  \captionof{figure}{Minimum stake lost in case $\ell$ committee members misbehave.}\label{minimum-stake-kusama}
\end{minipage}\hfill
\begin{minipage}[t]{.48\textwidth}
  \centering
  \resizebox{0.75\textwidth}{!}{\input{./plots_kusama/cost_of_buying.tex}}
  \captionof{figure}{Stake needed to replace committee
members; y-axis is logarithmic.}\label{cost-of-replacing-kusama}
\end{minipage}
\end{figure*}

We now give a brief description and experimental evaluation of the elections from the Kusama network. 
The network follows a very similar protocol to the one for Polkadot. 
However, there are a few differences. 
First, in Kusama, multiple elections are conducted each day. 
Because of this, we were able to collect  $1520$ elections conducted on the Kusama network. 
Moreover, in Kusama the maximum ballot length is $24$, so each voter can approve up to $24$ candidates (and not only up to $16$ as in Polkadot). 
The committee size in Kusama is $k=1000$, more than three times as large as in Polkadot. 

\paragraph{Description of the Data}
Analogous to \Cref{fig:basic1}, \Cref{fig:basic1-app} shows the size of the Kusama elections. Compared to the Polkadot elections, the Kusama elections have more candidates and fewer voters. 
Moreover, as for the Polkadot data, we see a steady increase in the number of voters  over time and a smaller fluctuation in the number of candidates. 
\bigskip

We also repeated our core experiments presented in the main body on the Kusama elections. 
However, due to the larger size and increased computation time for the Kusama elections (which have more candidates and most critically a much higher $k$), we only ran the experiments on a subset of $20$ elections, which are uniformly distributed over the available set of elections. 

\paragraph{Overlap between Committees}
\Cref{tab:overlapKus} depicts the average overlap of the committees computed by the different voting rules. 
As for Polkadot, the outcomes returned by the rules are typically very similar to each other (the difference is typically less than $10\%$). 
However, there are small differences in how the rules relate to each other here. 
In particular, the proportional rules show a slightly more diverse behavior and the separation between AV and SAV is less pronounced. 
In particular, while MES and seq-Phragmén still produce very similar committees, outcomes returned by seq-PAV are clearly closest to the ones returned by SAV. This different behavior of seq-PAV on the Kusama data is also reflected in some of our further experiments. 
Moreover, for Kusama, Phragmms is at a roughly similar distance from all other rules (including AV). 
Especially the similarity between Phragmms and AV is surprising here. 

\paragraph{Preventing Underrepresentation}
Analogous to \Cref{tab:overrep,app:rep_app}, \Cref{underrep-kusama} shows some measures regarding the prevention of underrepresentation. 
Moreover, as in \Cref{fig:representation3}, \Cref{average-sati-kusama} shows the minimum average satisfaction of $\ell$-supporting groups in our elections. 

The general trends are similar as for Polkadot, which is why we only point to the differences here, which mostly affect seq-PAV. 
Regarding the satisfaction of $\ell$-supporting groups, as for Polkadot, the proportional rules substantially outperform the best possible guarantee of $\ell-1$. 
The only differences are that the gap between seq-PAV and the other proportional rules is more pronounced and that SAV shows a better performance here.

Concerning priceability, seq-PAV shows a worse performance on the Kusama than on the Polkadot elections. 

Lastly, for both AV and SAV, the average number of candidates violating EJR+ is more than twice as high as for JR. This indicates that on the Kusama elections, there is a practical difference between these two notions of proportionality. 

\paragraph{Preventing Overrepresentation}

Analogous to \Cref{tab:overrep}, \Cref{tab:overrep-kusama} shows some measures related to the prevention of overrepresentation. 
The trends in the results are again similar as for Polkadot. 
The biggest difference is that seq-PAV performs worse with respect to the maximin support value and produces values closer to those of AV than those of Phragmms. 
Another (small) difference is that AV performs slightly better with respect to the exogenous cost of replacing $\frac{k}{3}$ and $\frac{k}{2}$ committee members (see also \Cref{cost-of-replacing-kusama} for the minimum exogenous cost of replacing committee members).

As \Cref{fig:overrep4}, \Cref{minimum-stake-kusama} depicts the minimum stake lost in Kusama, i.e., the minimum amount of stake assigned to a group of $\ell$ committee members in the maximin assignment. The results are very similar as for Polkadot, with the only difference being that seq-PAV performs slightly worse and is now more similar to AV and SAV than to the other proportional rules.

\end{document}